\documentclass[10pt]{article}

\usepackage[T1]{fontenc}										
\usepackage{lmodern}												
\usepackage{microtype}											
\usepackage[LGRgreek,frenchmath]{mathastext}
\usepackage[mathscr]{eucal}                 
\usepackage[sans]{dsfont}										
\usepackage{bbold}													
\usepackage{bbm}														
\usepackage{graphicx}                      	
\usepackage{mdwlist}												
\usepackage{natbib}    											
\setcitestyle{authoryear}
\usepackage[dvipsnames]{xcolor}							
\usepackage[plainpages=false, pdfpagelabels, backref=page]{hyperref} 
	\hypersetup{
		colorlinks   = true,
		citecolor    = RoyalBlue, 
		linkcolor    = RubineRed, 
		urlcolor     = Turquoise
	}
\usepackage[paperwidth=8.5in,paperheight=11.00in,top=1.00in, bottom=1.00in, left=0.75in, right=0.75in]{geometry}
\usepackage{mathtools}                      
\mathtoolsset{showonlyrefs=true}            
\usepackage{amsmath}
\usepackage{amssymb}
\usepackage{amsfonts}
\usepackage{amsthm}                         
\allowdisplaybreaks                         
\newtheoremstyle{plain}
  {}   				
  {}   				
  {\itshape}  
  {}       		
  {\mdseries\scshape} 
  {.}         
  { } 				
  {\thmname{#1}\thmnumber{ #2}\ifx#3\empty\else\ (#3)\fi}
\theoremstyle{plain}
\newtheorem{theorem}{\underline{Theorem}}
\newtheorem{lemma}[theorem]{\underline{Lemma}}       	
\newtheorem{proposition}[theorem]{\underline{Proposition}}

\newtheoremstyle{definition}
  {}   				
  {}   				
  {}  				
  {}      		
  {\mdseries\scshape} 
  {.}         
  { } 				
  {\thmname{#1}\thmnumber{ #2}\ifx#3\empty\else\ (#3)\fi}
\theoremstyle{definition}

\newtheorem{remark}[theorem]{\underline{Remark}}
\newtheorem{assumption}[theorem]{\underline{Assumption}}

\usepackage{sectsty}
\allsectionsfont{\mdseries\scshape}

\newcommand{\<}{\langle}
\renewcommand{\>}{\rangle}
\renewcommand{\(}{\left(}
\renewcommand{\)}{\right)}
\renewcommand{\[}{\left[}

\newcommand\Eb{\mathds{E}}
\newcommand\Fb{\mathds{F}}

\newcommand\Pb{\mathds{P}}

\newcommand\Rb{\mathds{R}}

\newcommand\Nb{\mathds{N}}

\newcommand\Ac{\mathscr{A}}

\newcommand\Fc{\mathscr{F}}

\newcommand\Om{\Omega}
\newcommand\sig{\sigma}

\newcommand\Gam{\Gamma}
\newcommand\lam{\lambda}
\newcommand\del{\delta}

\renewcommand\phi{\varphi}

\newcommand\Tb{\overline{T}}

\newcommand\Pbt{\widetilde{\Pb}}
\newcommand\Act{\widetilde{\Ac}}

\newcommand\mut{\widetilde{\mu}}
\newcommand\mt{\widetilde{m}}
\newcommand\st{\widetilde{s}}
\newcommand\St{\widetilde{S}}
\newcommand\Wt{\widetilde{W}}

\newcommand\Gamt{\widetilde{\Gam}}

\renewcommand\d{\partial}

\newcommand\ii{\mathtt{i}}
\newcommand\dd{\mathrm{d}}
\newcommand\ee{\mathrm{e}}

\newcommand{\Ib}[1]{\mathds{1}_{\{#1\}}}

\begin{document}

\title{Short-rate Derivatives in a Higher-for-Longer Environment}

\author{
Aram Karakhanyan
\thanks{Department of Mathematics, University of Edinburgh.  \textbf{e-mail}: \url{aram6k@gmail.com}}
\and
Takis Konstantopoulos
\thanks{Department of Mathematics, Uppsala University.  \textbf{e-mail}: \url{takiskonst@gmail.com}}
\and
Matthew Lorig
\thanks{Department of Applied Mathematics, University of Washington.  \textbf{e-mail}: \url{mlorig@uw.edu}}
\and 
Evgenii Samutichev
\thanks{Department of Applied Mathematics, University of Washington.
\textbf{e-mail}: \url{e.samutichev@gmail.com}}
}

\date{This version: \today}

\maketitle

\begin{abstract}
We introduce a class of short-rate models that exhibit a 	``higher for longer'' phenomenon.  Specifically, the short-rate is modeled as a general time-homogeneous one-factor Markov diffusion on a finite interval.  The lower endpoint is assumed to be regular, exit or natural according to boundary classification while the upper endpoint is assumed to be regular with absorbing behavior.  In this setting, we give an explicit expression for price of a zero-coupon bond (as well as more general interest rate derivatives) in terms of the transition density of the short-rate under a new probability measure, and the solution of a non-linear ordinary differential equation (ODE).  We then narrow our focus to a class of models for which the transition density and ODE can be solved explicitly.  For models within this class, we provide conditions under which the lower endpoint is regular, exit and natural.  Finally, we study two specific models -- one in which the lower endpoint is exit and another in which the lower endpoint is natural.  In these two models, we give an explicit solution of transition density of the short-rate as a (generalized) eigenfunction expansion.  We provide plots of the transition density, (generalized) eigenfunctions, bond prices and the associated yield curve.
\end{abstract}

%
%


\section{Introduction}
Between January of 2021 to December of 2021, the annual rate of inflation in the United States rose from 1.4\% to 7.0\% (source: \cite{USLaborStat}).
Although Treasury Secretary Janet Yellin initially called the rapid rise in inflation ``transitory,'' it had become apparent by early 2022 that high levels of inflation had become entrenched.  In an effort get inflation down to its 2\% target, between February of 2022 and August of 2023, the Federal Reserve raised its effective rate from roughly 0.1\% to just over 5.3\% (source: \cite{StLouisFed}).  Despite the historically fast pace of interest rate hikes, the annual rate of inflation in September of 2023 was still sitting at 3.7\%, leading Federal Reserve Chairman Jerome Powell to declare that interest rates would likely remain ``higher for longer.'' And, indeed, after topping out 5.32\% in July of 2023, the federal overnight rate remained at this value until August of 2024. 
\\[0.5em] 
In this paper, we present a class of one-factor short-rate models that capture this ``higher for longer'' phenomenon.  Specifically, we assume that the short-rate is modeled as a Markov diffusion on a finite interval. The upper end of the interval is a regular (i.e., accessible) boundary according the Feller's boundary classification, and is assumed to have absorbing behavior.  That is, when the short rate reaches the upper end of the interval, it remains there forever.  Although in reality, the federal reserve will likely lower interest rates when inflation returns to the Federal Reserve's 2\% target, the absorbing behavior captures the effect of interest rates remaining high for an extended period of time.  Thus, the class of models in this paper can be seen as valid up to some finite time-horizon, after which the short-rate may be lowered.  The class of models in this paper has the additional benefit of being analytically tractable, as the bond price can be written explicitly.
\\[0.5em]
The rest of this paper proceeds as follows:
in Section \ref{sec:model} we introduce a class of short-rate models described by a one-factor Markov diffusion on a finite interval.  
In this setting, we use risk-neutral pricing to derive an expression for the value of a financial derivative whose payoff depends on the terminal value of the short-rate.
We additionally relate the value of this financial derivative to the solution of a partial differential equation (PDE).
In Section \ref{sec:pde}, we perform a change of variables that enables us to express the solution of the pricing PDE in terms of the solutions of a nonlinear ordinary differential equation (ODE) and generalized eigenvalue problem.
Lastly, in Section \ref{sec:example}, we consider a class of short-rate models for which the ODE and generalized eigenvalue problem described in Section \ref{sec:pde} can be solved explicitly.  We additionally perform a numerical study of bond prices and the resulting yield curves for two specific examples within this class.  The two examples are qualitatively different in that the first example has an exit boundary at the origin and the eigenvalue problem has a discrete set of solutions, while the second example has a natural boundary at the origin and the eigenvalue problems has a continuous set of solutions.

%
%

\section{Short-rate model and Pricing}
\label{sec:model}
To begin, we fix a time horizon $\Tb < \infty$ and consider a continuous-time financial market, defined on a filtered probability space $(\Om,\Fc,\Fb,\Pb)$ with no arbitrage and no transaction costs.  The probability measure $\Pb$ represents the market's chosen pricing measure taking the \textit{money market account} $M = (M_t)_{0 \leq t \leq \Tb}$ as num\'eraire.  The filtration $\Fb = (\Fc_t)_{0 \leq t \leq \Tb}$ represents the history of the market.
\\[0.5em]
We suppose that the money market account $M$ is strictly positive, continuous and non-decreasing.  As such, there exists a non-negative $\Fb$-adapted \textit{short-rate} process $X = (X_t)_{0 \leq t \leq \Tb}$ such that
\begin{align}
\dd M_t
	&=	X_t M_t \, \dd t ,  &
M_0
	&> 0 . \label{eq:dM}
\end{align}
We will focus on the case in which the dynamics of the short-rate $X$ are described by a Stochastic Differential Equation (SDE) of the following form
\begin{align}
\dd X_t
	&= \Ib{X_t \in (0,L)} \Big( \mu(X_t) \dd t + \sig(X_t) \dd W_t \Big) , &
X_0
	&\in [0,L] , \label{eq:short-rate-SDE}
\end{align}
where $W$ is a $(\Pb,\Fb)$-Brownian motion.  We assume that the \textit{drift} $\mu : (0,L) \to \Rb$ and \textit{volatility} $\sig : (0,L) \to \Rb_+$ are such that the origin is either a regular, exit or a natural boundary, and that $L$ is a regular boundary, according to Feller's boundary classification for scalar diffusions \cite{Feller1952THEPD}. 
We further assume that $X$ has a unique strong solution up until the first exit time of $(0,L)$.  Observe that, as soon as $X$ exits $(0,L)$, we have $\dd X_t = 0$ due to the indicator function on the right-hand side of \eqref{eq:short-rate-SDE}.  And thus, the endpoint $\{ 0 \}$ is an absorbing state if it is a regular or exit boundary and endpoint $\{L\}$ is an absorbing state.
\\[0.5em]
Let us denote by $V = (V_t)_{0 \leq t \leq T}$ the value of a financial derivative that has a payoff $g(X_T)$ at time $T \leq \Tb$, where $g: [0,L] \to \Rb$.  Using standard risk-neutral pricing, we have
\begin{align}
\frac{V_t}{M_t}
	&=	\Eb \Big( \frac{V_T}{M_T} \Big| \Fc_t \Big) , &
0
	&\leq t \leq T .
\end{align}
Solving for $V_t$ and using $M_t = M_0 \exp ( \int_0^t X_s \dd s )$ and $V_T = g(X_T)$, we obtain
\begin{align}
V_t
	&=	\Eb \Big( \ee^{- \int_t^T X_s \dd s} g(X_T) \Big| \Fc_t \Big) =: u(t,X_t) , \label{eq:u-def}
\end{align} 
where we have used the Markov property of $X$ to express $V_t$ as a function $u:[0,T] \times [0,L] \to \Rb$ of $(t,X_t)$.  By the Feynman-Kac formula, the function $u$ satisfies the following partial differential equation (PDE) and boundary conditions (BCs)
\begin{align}
( \d_t + \Ac - x ) u
	&=	0 , &
&u(T,x)
	=	g(x) , \label{eq:u-pde} \\
u(t,0)
	&=	g(0) , &
	&\text{if $0$ is regular or exit} \label{eq:u-bc-0}, \\
u(t,L)
	&=	\ee^{- L (T - t) } g(L) , \label{eq:u-bc-L} 
\end{align}
where the operator $\Ac$ is the \textit{infinitesimal generator of $X$} under $\Pb$ and is given explicitly by
\begin{align}
\Ac
	&=	\mu(x) \d_x + \tfrac{1}{2} \sig^2(x) \d_x^2  . \label{eq:A}
\end{align}
Throughout this paper we assume that Cauchy problem \eqref{eq:u-pde}-\eqref{eq:u-bc-L} has a unique classical solution.
\\[0.5em]
Notice that the function $u$ defined in \eqref{eq:u-def} depends on the maturity date $T$ and the payoff function $g$.  If we wish to emphasize the dependence of $u$ on $T$ and $g$, we may sometimes write $u(t,X_t) \equiv u(t,X_t;T,g)$.  For example, denoting by $B^T = (B_t^T)_{0 \leq t \leq T}$ the price of a \textit{zero-coupon bond} that pays one unit of currency at the maturity date $T$ and by $Y^T = (Y_t^T)_{0 \leq t \leq T}$ the associated \textit{yield to maturity}, we have
\begin{align}
B_t^T
	&= u(t, X_t; T, 1) , &
Y_t^T 
	&= \frac{- \log B_t^T}{T-t}  . \label{eq:bond-yield}
\end{align}

%
%

\section{Solving the pricing PDE}
\label{sec:pde}
In this section, we will obtain an explicit expression for the solution $u$ of Cauchy problem \eqref{eq:u-pde}-\eqref{eq:u-bc-L}.  
The expression we obtain will depend on the solution of a first order, nonlinear, ordinary differential equation (ODE) as well as the (possibly improper) eigenfunctions and eigenvalues of a linear second-order differential operator.  We begin with a short lemma.

\begin{lemma}
Let $u$ be the unique classical solution of Cauchy problem \eqref{eq:u-pde}-\eqref{eq:u-bc-L}.  
Suppose $f$ satisfies the following ordinary differential equation
\begin{align}
0
	&=	\tfrac{1}{2} \sig^2(x) (\d_x f)^2 - \Ac f(x) - x ,  \label{eq:f-ode}
\end{align}
where the operator $\Ac$ is given by \eqref{eq:A}.
Define the function $w:[0,T] \times [0,L] \to \Rb$ as follows
\begin{align}
w(t,x) 
	&:=	\ee^{f(x)} \Big( u(t,x) - \ee^{-x(T-t)} g(x)  \Big) . \label{eq:w-def}
\end{align}
Then the function $w$ satisfies the following Cauchy problem
\begin{align}
( \d_t + \Act ) w + q
	&=	0, &
&w(T,x)
	=	0, \label{eq:w-pde} \\
w(t,0)
	&=	0, &
	&\text{if $0$ is regular or exit} , \label{eq:w-bc-0} \\
w(t,L)
	&=	0 , \label{eq:w-bc-L} 
\end{align}
where the operator $\Act$ and the function $q$ are given by
\begin{align}
\Act
	&=	\mut(x) \d_x + \tfrac{1}{2} \sig^2(x) \d_x^2 &
\mut(x)
	&=	\mu(x) - \sig^2(x) f'(x) , &
q(t,x)
	&=	( \d_t + \Act ) \ee^{-x(T-t) + f(x)} g(x) . \label{eq:q}
\end{align}
\end{lemma}

\begin{proof}
First, we have from \eqref{eq:u-pde}, \eqref{eq:u-bc-0} \eqref{eq:u-bc-L} and \eqref{eq:w-def} that
\begin{align}
w(T,x) 
	&:=	\ee^{f(x)} \Big( u(T,x) - \ee^{-x(T-T)} g(x) \Big)
	=		\ee^{f(x)} \Big( g(x) - \ g(x) \Big) 
	=		0 , \\
w(t,0) 
	&:=	\ee^{f(0)} \Big( u(t,0) - \ee^{-0(T-t)} g(0) \Big)
	=		\ee^{f(0)} \Big( g(0) - g(0)\Big)
	=		0 , \\
w(t,L) 
	&:=	\ee^{f(L)} \Big( u(t,L) - \ee^{-L(T-t)} g(L) \Big)
	=		\ee^{f(L)} \Big( \ee^{- L (T - t)} g(L) - \ee^{-L(T-t)} g(L)\Big)
	=		0 ,
\end{align}
which establishes the BCs in \eqref{eq:w-pde}, \eqref{eq:w-bc-0} and \eqref{eq:w-bc-L}.  Next, we have from \eqref{eq:w-def} that
\begin{align}
u(t,x)
	&=	\ee^{-f(x)} \Big( w(t,x) + \ee^{-x(T-t) + f(x)} g(x) \Big) . \label{eq:u-w}
\end{align}
Inserting \eqref{eq:u-w} into \eqref{eq:u-pde}, multiplying both sides by $\ee^{f(x)}$ and using the fact that $f$ satisfies \eqref{eq:f-ode}, we obtain
\begin{align}
0
	&=	\ee^{f(x)} ( \d_t + \Ac - x ) \ee^{-f(x)} \Big( w(t,x) + \ee^{-x(T-t) + f(x)} g(x) \Big) \\
	&=	\ee^{f(x)} \ee^{-f(x)}  ( \d_t + \Act ) \Big( w(t,x) + \ee^{-x(T-t) + f(x)} g(x) \Big) \\
	&=	( \d_t + \Act ) w(t,x) + q(t,x) .
\end{align}
which establishes the PDE in \eqref{eq:w-pde}.
\end{proof}

\noindent At this point, it will be convenient to introduce the \textit{fundamental solution} $\Gamt$ associated with the operator $(\d_t + \Act)$.  Specifically, for any $(T,y) \in [0,\Tb] \times [0,L]$ we define $\Gamt(\, \cdot \, , \, \cdot \, ; T, y) : [0,T] \times [0,L] \to \Rb_+ $ as the unique classical solution of
\begin{align}
(\d_t + \Act) \Gamt(\, \cdot \, , \, \cdot \, ;T,y)
	&=	0 , &
&\Gamt(T,x;T,y)
	=	\del_y(x) , \label{eq:Gamma-pde} \\
\Gamt(t,0;T,y)
	&=	0 , &
	&\text{if $0$ is regular or exit} , \label{eq:Gamma-bc-0} \\
\Gamt(t,L;T,y)
	&=	0 . \label{eq:Gamma-bc-L}
\end{align}

\begin{remark}
\label{rem:equivprob}
Consider the following change of measure
\begin{align}
\frac{\dd \Pbt}{\dd \Pb}
	&:= \exp \Big( \frac{1}{2} \int_{0}^{\Tb}{\sig^2(X_t) (f'(X_t))^2 \dd t} - \int_{0}^{\Tb} \sig(X_t) f'(X_t) \dd W_t \Big) . 
			\label{eq:radon-nikodym}
\end{align}
We have from Girsanov's Theorem that the process $\Wt = (\Wt_t)_{0 \leq t \leq \Tb}$, defined by 
\begin{align}
\dd \Wt_t 
	&:= \sig(X_t) f'(X_t) \dd t + \dd W_t , &
\Wt_0
	&=	0 , \label{eq:W-tilde}
\end{align}
is a $(\Pbt,\Fc)$-Brownian motion.  From \eqref{eq:short-rate-SDE} and \eqref{eq:W-tilde}, we see that the dynamics of $X$ under $\Pbt$ are 
\begin{align}
\dd X_t
	&=	\Ib{X_t \in (0,L)} \Big( \mut(X_t) \dd t + \sig(X_t) \dd W_t \Big) ,
\end{align}
and the generator of $X$ under $\Pbt$ is the operator $\Act$ given in \eqref{eq:q}.  If follows that the solution $\Gamt$ of \eqref{eq:Gamma-pde}-\eqref{eq:Gamma-bc-L} is the \textit{transition density} of $X$ under $\Pbt$.
\end{remark}	

\noindent 
We can now state the main result of this section.
\begin{theorem}
The solution $u$ of Cauchy problem \eqref{eq:u-pde}-\eqref{eq:u-bc-L} is given by
\begin{align}
u(t,x)
	&=	\ee^{-x(T-t)} g(x) + \ee^{-f(x)} \int_t^T \dd s \int_0^L \dd \xi \, \Gamt(t,x;s,\xi) ( \d_s + \Act ) \ee^{-\xi(T-s) + f(\xi)} g(\xi) , \label{eq:u-explicit}
\end{align}
where the operator $\Act$ is given by \eqref{eq:q} and the functions $f$ and $\Gamt$ satisfy \eqref{eq:f-ode} and \eqref{eq:Gamma-pde}-\eqref{eq:Gamma-bc-L}, respectively.
\end{theorem}

\begin{proof}
By Duhamel's principle, we have from \eqref{eq:w-pde} that
\begin{align}
w(t,x)
	&=	\int_t^T \dd s \int_0^L \dd \xi \, \Gamt(t,x;s,\xi) q(s,\xi)  , 
\end{align}
where $\Gamt$ satisfies \eqref{eq:Gamma-pde}-\eqref{eq:Gamma-bc-L}. Thus, we have from \eqref{eq:u-w} that
\begin{align}
u(t,x)
	&=	\ee^{-x(T-t)} g(x) + \ee^{-f(x)} \int_t^T \dd s \int_0^L \dd \xi \, \Gamt(t,x;s,\xi) q(s,\xi) . \label{eq:u-2}
\end{align}
Equation \eqref{eq:u-explicit} follows by inserting expression \eqref{eq:q} for $q$ into \eqref{eq:u-2}.
\end{proof}

\noindent
Observe from \eqref{eq:u-2} that, to write the solution $u$ of \eqref{eq:u-pde}-\eqref{eq:u-bc-L}, we need an expression for a solution $f$ of \eqref{eq:f-ode}.  The following Proposition gives sufficient conditions on the coefficients $\mu$ and $\sig$ under which \eqref{eq:f-ode} has an explicit solution.

\begin{proposition}
Suppose $\mu$ and $\sig$ are given by
\begin{align}
\mu(x)
	&=	\frac{ - \sig^2(x) \phi'(x) }{4 \phi(x)} , &
\sig(x)
	&=	\sqrt{ \frac{2x}{\phi(x)} } , \label{eq:mu-sig-conditions}
\end{align}
for some non-negative function $\phi$.  Then
\begin{align}
f(x)
	&:=	- \int \sqrt{\phi(x)} \dd x + c , \label{eq:f-explicit}
\end{align}
is a solution of \eqref{eq:f-ode}, where $c$ is an arbitrary constant.  
\end{proposition}

\begin{proof}
Define $F := - f'$.  Then we have from \eqref{eq:f-ode} that
\begin{align}
0
	&=	F' + F^2 + \frac{2\mu}{\sig^2} F - \frac{2x}{\sig^2} 
	=	F' + F^2 - \frac{\phi'}{2\phi} F - \phi , \label{eq:F-special}
\end{align}
where the second equality follows from \eqref{eq:mu-sig-conditions}.  It is easy to see by direction substitution that $F = \sqrt{\phi}$ satisfies \eqref{eq:F-special}.  Thus, we have from $f = - \int F \dd x$ that $f$ is given by \eqref{eq:f-explicit}.
\end{proof}

\noindent
We note from \eqref{eq:u-2} that, to write the solution $u$ of \eqref{eq:u-pde}-\eqref{eq:u-bc-L}, we need an expression for the solution $\Gamt$ of \eqref{eq:Gamma-pde}-\eqref{eq:Gamma-bc-L}.  To this end, it will be helpful to write the operator $\Act$ in \textit{self-adjoint form}.  We have
\begin{align}
\Act
	&=		\frac{1}{\mt(x)} \d_x \frac{1}{\st(x)} \d_x , &  
\st(x)
	&:=	\exp \Big( - \int \dd x \, \frac{2 \mut(x)}{\sig^2(x)} \Big) , &
\mt(x)
	&:=	\frac{2}{\sig^2(x)} \exp \Big( \int \dd x \, \frac{2 \, \mut(x)}{\sig^2(x)} \Big) , \label{eq:scale-speed}
\end{align}
where we have introduced the \textit{scale density} $\st$ and the \textit{speed density} $\mt$ of the operator $\Act$.
The scale and speed densities will be needed later in this paper to determine the behavior of the short-rate $X$ at the origin.

\begin{assumption}
\label{ass:simple}
Henceforth, we assume the spectrum of $\Act$, denoted $\sig(\Act)$, is simple.
\end{assumption}

\noindent
Now, consider the following \textit{eigenvalue problem}
\begin{align}
\Act \psi_\lam 
	&= \lam \psi_\lam , \label{eq:psi-ode}  \\
\psi_\lam(0) 
	&= 0, &
	&\text{if $0$ is regular or exit} \label{eq:psi-bc-0} , \\
\psi_\lam (L)
	&= 0. \label{eq:psi-bc-L} 
\end{align} 
In general, the spectrum of self-adjoint operator $\Act$ has the decomposition $\sig(\Act) = \sig_d(\Act) \cup \sig_c(\Act)$, where $\sig_d(\Act)$ and $\sig_c(\Act)$ denote the \textit{discrete} and \textit{continuous} portions of $\sig(\Act)$, respectively. We will denote by $(\lam_n)_{n \in \Nb} = \sig_d(\Act)$ the set of \textit{proper eigenvalues of $\Act$} and by $(\lam(\rho))_{\rho \in \Rb_+} = \sig_c(\Act)$ the set of \textit{improper eigenvalues of $\Act$}.
Similarly, we will denote by $( \psi_n)_{n \in \Nb}$ the \textit{proper eigenfunctions of $\Act$}, normalized as follows
\begin{align}
\psi_n 
	&:= \psi_{\lam_n} , &
\lam_n
	&\in \sig_d(\Ac) , &
\< \psi_n , \psi_k \>_{\mt}
	&:=	\int_0^L \psi_n(x) \psi_k(x) \mt(x) \dd x = \delta_{n, k} , \label{eq:normalization}
\end{align}
and by $( \psi(\rho, \, \cdot \, ))_{\rho \in \Rb_+}$ the \textit{improper eigenfunctions of $\Act$}, normalized according to
\begin{align}
\psi(\rho, \, \cdot \, )
	&:= \psi_{\lam(\rho)} , &
\lam(\rho)
	&\in \sig_c(\Ac) , &
\int_0^L \psi(\rho, x) \psi(\rho',x) \mt(x) \dd x 
	&= 	\delta(\rho-\rho') , \label{eq:normalization-continuous}
\end{align}
where $\lam(\rho) = -c \rho^2$ for some $c > 0$. 
Then, under Assumption \ref{ass:simple}, we have from \cite[Chapter 5]{hansonyakovlev} that the solution $\Gamt$ of \eqref{eq:Gamma-pde}-\eqref{eq:Gamma-bc-L} has the following (generalized) eigenfunction representation
\begin{align}
\Gamt(t,x;T,y)
	&=	 \mt(y) \Big( \sum_{n \in \Nb} \ee^{\lam_n(T-t)} \psi_{n}(x) \psi_{n}(y) 
			+  \int_{0}^{\infty}{\ee^{\lam(\rho) (T-t)} \psi(\rho,x)\psi(\rho, y) \dd \rho} \Big) .
			\label{eq:Gamma}
\end{align}
Thus, we need only to solve eigenvalue problem \eqref{eq:psi-ode}-\eqref{eq:psi-bc-L} to construct the transition density $\Gamt$.
\noindent
The following proposition, which is based on the transformation to \textit{Liouville normal form} \cite[Chapter 1]{NIST:DLMF}, can be helpful to find solutions of the eigenvalue equation \eqref{eq:psi-ode} in the case where $\mu$ and $\sig$ are given by \eqref{eq:mu-sig-conditions} for some non-negative function $\phi$.

\begin{proposition}
Let $\phi$ be a non-negative function.  Suppose $\mu$ and $\sig$ are given by \eqref{eq:mu-sig-conditions}.  Then the scale density $\st(x)$, the speed density $\mt(x)$, defined in \eqref{eq:scale-speed}, are  given by 
\begin{align}
\st(x)
	&=	C \sqrt{ \phi(x) } \exp \Big( -2 \int \dd x \sqrt{\phi(x)} \Big) , &
\mt(x)
	&=	\frac{\sqrt{\phi(x)}}{ C x} \exp \Big( 2 \int \dd x \sqrt{\phi(x)} \Big) , \label{eq:speed-density} 
\end{align}
where $C$ is an arbitrary constant, and the operator $\Act$ is given by
\begin{align}
\Act
	&=	\Big( \frac{2 x }{\sqrt{\phi (x)}} - \frac{x \phi'(x)}{2 \phi^2 (x)} \Big) \d_x + \frac{x}{\phi (x)} \d_x^2 , \label{eq:Act-again}
\end{align}
Suppose $\psi_{\lam}$ is a (possibly improper) eigenfunction of $\Act$, which satisfies \eqref{eq:psi-ode}-\eqref{eq:psi-bc-L}.  Define functions $p$ and $\eta_{\lam}$ through the following equations
\begin{align}
p(x)
	&:=	\int \dd x \frac{1}{\sig(x)} , &
\psi_{\lam}(x)
	&=	\eta_{\lam}(p(x)) \sqrt{\sig(x) \st(x)} , \label{eq:z-def}
\end{align}
Then $\eta_{\lam}$ satisfies
\begin{align}
\tfrac{1}{2} \eta_{\lam}'' - V \eta_{\lam}
	&=	\lam \eta_{\lam} , &
V(z)
	&:=	U(p^{-1}(z)) , &
U(x)
	&:=	\frac{\phi'(x)}{8 \phi^2(x)}+ \frac{3 }{16 x \phi(x)} + x , \label{eq:eta-ode} \\
\eta_{\lam}(p(0))
	&=	0 , &
	&\text{if $0$ is regular or exit} , \label{eq:eta-bc-0} \\
\eta_{\lam}(p(L))
	&=	0 , \label{eq:eta-bc-L}
\end{align}
where $p^{-1}$ is the inverse of $p$. 
\end{proposition}

\begin{proof}
The expressions for $\mt$, $\st$ and $\Act$ can be computed directly from \eqref{eq:q}, \eqref{eq:mu-sig-conditions}, \eqref{eq:f-explicit} and \eqref{eq:scale-speed}.  Next, using \eqref{eq:psi-ode} and \eqref{eq:z-def}, we have 
\begin{align}
0
	&=	\frac{1}{\sqrt{\sig(x) \st(x)}} ( \Act - \lam) \psi_\lam(x) 
	 =	\tfrac{1}{2} \eta_\lam''(p(x)) - U(x) \eta_\lam(p(x)) - \eta_\lam(p(x)) .
\end{align}
Setting $x = p^{-1}(z)$, we obtain \eqref{eq:eta-ode}.  The BCs \eqref{eq:eta-bc-0}-\eqref{eq:eta-bc-L} for $\eta_\lam$ follow from the BCs \eqref{eq:psi-bc-0}-\eqref{eq:psi-bc-L} for $\psi_\lam$ and \eqref{eq:z-def}.
\end{proof}

%
%

\section{A class of analytically tractable models}
\label{sec:example}
In this section, we analyze a class of models for the short rate $X$ in which the drift $\mu$ and volatility $\sig$ are given by \eqref{eq:mu-sig-conditions} with the function $\phi$ being given by
\begin{align}
\phi(x) 
	&= \frac{2}{a^2} x^{2k-1}, & 
a 
	&> 0 . \label{eq:model-phi}
\end{align}
Inserting \eqref{eq:model-phi} into \eqref{eq:mu-sig-conditions} we obtain
\begin{align}
\sig(x)
	&=	a x^{1-k} , &
\mu(x)
	&=	a^2 \Big( \tfrac{1}{4} - \tfrac{k}{2} \Big) x^{1 - 2k}. \label{eq:mu-sig-example}
\end{align}
Thus, we have from \eqref{eq:short-rate-SDE} and \eqref{eq:mu-sig-example} that the dynamics of the short-rate $X=(X_t)_{0 \leq t \leq \Tb}$ are given by
\begin{align}
\dd X_t
		&= \Ib{X_t \in (0,L)} \Big( a^2 (\tfrac14 - \tfrac{k}{2})X_t^{1-2k} \dd t + a X_t^{1-k}\dd W_t \Big) , &
X_0
		&\in [0,L] , \label{eq:our-short-rate}
\end{align}
\begin{remark}
$k = 0$ corresponds to a geometric Brownian motion with $\mu = \frac{a^2}{4}$ and $\sig = a$. 
\end{remark}
\noindent It then follows from \eqref{eq:Act-again} that the operator $\Act$ is given by
\begin{align}
\Act
	&=	\Big( a^2 \big( \tfrac{1}{4} - \tfrac{k}{2} \big) x^{1-2 k} + a \sqrt{2 x^{3-2k}} \Big) \d_x + \frac{1}{2}a^2 x^{2-2k} \d_x^2 , \label{eq:Act-example}
\end{align}
and from \eqref{eq:speed-density} the scale density $\st$ and the speed density $\mt$ are given by
\begin{align}
\st(x) 
	&= \begin{cases}
			Cx^{k-\frac12} \exp\left(-\frac{2\sqrt2}{a(k+\frac12)}x^{k+\frac12}\right), &k \neq -\frac12 \\
			Cx^{-1-\frac{2\sqrt2}{a}}, &k = -\frac12 
		\end{cases} , \label{eq:st} \\
\mt(x) 
	&= \begin{cases}
			\frac{2}{Ca^2} x^{k-\frac32} \exp\left(\frac{2\sqrt{2}}{a(k+\frac12)}x^{k+\frac12}\right), &k \neq -\frac12 \\
			\frac{2}{Ca^2}x^{-2+\frac{2\sqrt2}{a}}, &k =-\frac12
		\end{cases} , \label{eq:mt}
\end{align}
In order to write spectral representation of the transition density \eqref{eq:Gamma} for the short-rate \eqref{eq:our-short-rate}, we must determine the structure of the spectrum $\sig(\Act)$ of the diffusion generator $\Act$ \eqref{eq:Act-example}, which depends on whether the origin is regular, exit or natural. This will also determine if a BC \eqref{eq:Gamma-bc-0} is needed at the origin. 
\begin{proposition}
\label{thm:origin}
Fix an interval $(0,L)$ and consider the operator $\Act$ given by \eqref{eq:Act-example}.  
For $k \in (-\infty,0]$ the origin is natural, for $k \in (0,1/2]$ the origin is exit, and for $k \in (1/2,\infty)$ the origin is regular.
\end{proposition}
\begin{proof}
See Appendix \ref{sec:origin}.
\end{proof}
\noindent Using \eqref{eq:z-def} and \eqref{eq:eta-ode} we have for $k \neq 0$ that
\begin{align}
z = p(x)
	&=	x^k/ak , &
U(x)
	&=	\frac{1}{32} x^{-2 k} \(32 x^{2 k+1}+a^2 (4 k+1)\) , &
V(z) 
	&=	(akz)^{1/k} + \frac{1+4k}{32 k^2 z^2} ,\label{eq:Liouville-transformation}
\end{align}
and for $k = 0$ we have
\begin{align}
z = p(x) 
	&= \frac{1}{a}\log x , & 
U(x) 
	&= x + \frac{a^2}{32} , &
V(z) 
	&= \ee^{z} + \frac{a^2}{32}. \label{eq:Liouville-transformation-0}
\end{align}
Thus, from \eqref{eq:eta-ode}, \eqref{eq:eta-bc-0} and \eqref{eq:eta-bc-L}, we have for $k \neq 0$ that $(\eta_\lambda,\lambda)$ satisfy 
\begin{align}
\tfrac{1}{2} \eta_{\lambda}''(z) - \Big( (akz)^{1/k} + \frac{1+4k}{32 k^2 z^2} \Big) \eta_\lambda(z)
	&=	\lam \eta_\lambda(z) , \label{eq:eta-ode-2} \\
\eta_\lambda(0^k / ak)
	&=	0 , &
	&\text{if $0$ is regular or exit} , \label{eq:eta-bc-2-0} \\
\eta_\lambda(L^k / ak)
	&=	0  . \label{eq:eta-bc-2-L}
\end{align}
For $k = 0$ the origin is a natural boundary by Proposition \ref{thm:origin}.  Hence, we have
\begin{align}
	\tfrac{1}{2} \eta_\lambda''(z) - \Big( \ee^{z} + \frac{a^2}{32}\Big) \eta_\lambda(z)
		&=	\lam \eta_\lambda(z) , \label{eq:eta-ode-2-0} \\
	\eta_\lambda(\tfrac{\log L}{a})
		&=	0  . \label{eq:eta-bc-L-0}
\end{align}
Notice that when $k \leq 0$, after the Liouville transformation $z=p(x)$ as in \eqref{eq:Liouville-transformation} and \eqref{eq:Liouville-transformation-0}, the origin turns into $-\infty$, which plays a role in the following decomposition of spectrum $\sig(\Act)$.
\begin{proposition}
\label{thm:spectrum}
Fix an interval $(0,L)$ and consider the operator $\Act$ given by \eqref{eq:Act-example}. 
For $k > 0$ the spectrum is purely discrete, i.e. $\sig(\Act) = \sig_d(\Act) \subset (-\infty, 0)$, for $k < 0$ the spectrum is purely continuous $\sig(\Act) = \sig_c(\Act) = (-\infty, 0]$, and for $k=0$ the spectrum is mixed, with discrete portion $\sig_d(\Act) \subset (-\tfrac{a^2}{32}, 0]$ clustering at $-\tfrac{a^2}{32}$ and continuous portion $\sig_{c}(\Act) = (-\infty, -\tfrac{a^2}{32}]$.
\end{proposition}
\begin{proof}
Notice that the right endpoint $L$ is always regular. By Proposition \ref{thm:origin} in case $k > 0$ the origin is either exit or regular. According to \cite{linetskybook} regular and exit boundaries are always \textit{non-oscillatory}.  Therefore, both endpoints in this case are non-oscillatory and the operator $\Act$ is of the \textit{spectral category I} (i.e., its spectrum is purely discrete and lies in $(-\infty, 0]$).
\\[0.5em]
For $k \leq 0$, the origin is natural by Proposition \ref{thm:origin}. From \eqref{eq:Liouville-transformation} and \eqref{eq:Liouville-transformation-0} we have 
\begin{align}
	\lim_{x \to 0+}{p(x)} = -\infty,
\end{align}
i.e., the origin is transformed into $-\infty$ by Liouville transformation \eqref{eq:z-def}. Then, we investigate the potential $U(x)$ in \eqref{eq:Liouville-transformation} and \eqref{eq:Liouville-transformation-0}.  We have
\begin{align}
\Lambda &:=
	\lim_{x \to 0}{U(x)} = \begin{cases}
		\frac{a^2}{32}, &k = 0, \\
		0, &k < 0.
\end{cases}
\end{align}
As this limit is finite, the origin is an \textit{oscillatory} boundary with cutoff $-\Lambda$. This, together with $L$ being regular and thus non-oscillatory, implies that the operator $\Act$ is of the \textit{spectral category II}. Therefore for $\Lambda = 0$ ($k < 0$) it has purely continuous spectrum in $(-\infty, 0]$, while for $\Lambda = \frac{a^2}{32}$ ($k = 0$) it has continuous spectrum $\sig_c(\Act) = (-\infty, -\Lambda] = (-\infty, -\tfrac{a^2}{32}]$ and discrete spectrum in $(-\tfrac{a^2}{32}, 0]$ clustering at $-\tfrac{a^2}{32}$.
\end{proof}
\noindent For certain choices of $k$, we can express $\eta_\lam$, the solution of \eqref{eq:eta-ode-2} corresponding to eigenvalue $\lam$ in terms of special functions. We provide some examples here
\begin{align}
k
	&=	1/2 , &
\eta_\lam(z)
	&= \frac{\ee^{-\frac{az^2}{2\sqrt2}}}{\sqrt{z}} \Big(c_1 U\Big(\tfrac{\lambda}{\sqrt2 a}, 0; \tfrac{az^2}{\sqrt2}\Big) + c_2 L_{-\frac{\lambda}{\sqrt2 a}}^{(-1)}\Big(\tfrac{az^2}{\sqrt2}\Big)\Big) , \label{eq:eta-k12-solution} \\
k 
	&= 0 , & 
\eta_\lam (z)
	&= c_1 I_{-\tfrac{1}{2}\sqrt{a^2+32\lambda}}(2\sqrt{2} \ee^{z/2}) + c_2 I_{\tfrac{1}{2}\sqrt{a^2+32\lambda}}(2\sqrt{2} \ee^{z/2}) , \\  
k
	&=	-1/4 , &
\eta_\lam(z)
	&= c_1 z \ee^{-\frac{ \sqrt2}{z}\left(\frac{16}{a^2}+z^2 \ii \sqrt{-\lambda}  \right)} \text{HeunD}\Big(\tfrac{64}{a^2} \ii \sqrt{-\lambda},-2 \ii\sqrt{-2\lambda },\tfrac{32}{a^2}
	\sqrt{2},2,-2 \ii \sqrt{-2\lambda };z\Big) \\ &&
	&\quad + c_2 z \ee^{\frac{ \sqrt2}{z}\left(\frac{16}{a^2}+z^2  \ii \sqrt{-\lambda}  \right)}\text{HeunD}\Big(\tfrac{64}{a^2} \ii \sqrt{-\lambda},2 \ii\sqrt{-2\lambda },-\tfrac{32}{a^2}
	\sqrt{2},2,2 \ii \sqrt{-2\lambda };z\Big), \\
k
	&= -1/2 , &
\eta_\lam(z)
	&= \sqrt{z} \Big(c_1 J_{\frac{2 \sqrt2}{a}}\Big(-\ii \sqrt{2\lambda} z\Big) + c_2 Y_{\frac{2\sqrt2}{a}}\Big(-\ii \sqrt{2\lambda} z \Big) \Big), \label{eq:eta-k-neg12-solution}
\end{align}
where $U(\alpha, \beta; z)$ is the Tricomi's confluent hypergeometric function, $L^{(\alpha)}_n(z)$ is the generalized Laguerre polynomial, $\text{HeunD}(q,\alpha,\gamma,\delta,\epsilon;z)$ is the double-confluent Heun function \cite[HeunD function]{HeunD}, $J_\alpha(z)$ and $Y_\alpha(z)$ are Bessel functions, and $I_\alpha(z)$ is the modified Bessel function of the first kind. The constants $c_1$ and $c_2$ are chosen to satisfy the BCs \eqref{eq:eta-bc-0}, \eqref{eq:eta-bc-L} and \eqref{eq:eta-bc-L-0} together with normalization \eqref{eq:normalization} or \eqref{eq:normalization-continuous} for \eqref{eq:z-def} with respect to $\langle \cdot, \cdot \rangle_{\mt}$ where $\mt$ is the speed density given by \eqref{eq:speed-density}. 

\subsection{Example: $k = 1/2$, the origin is exit}
\label{subsec:k12}
For $k=1/2$, the SDE \eqref{eq:our-short-rate} of the short-rate $X$ becomes 
\begin{align}
\dd X_t
	&= \Ib{X_t \in (0,L)} a \sqrt{X_t}\dd W_t . \label{eq:k12-SDE}
\end{align}
In Figure \ref{fig:k12-trajectories}, using a standard Euler-Maruyama method, we plot four sample trajectories of the short-rate \eqref{eq:k12-SDE} under $\Pb$ (corresponding to generator $\Ac$ in \eqref{eq:A}) as well as the corresponding trajectories of the short-rate $X$ under $\Pbt$ (corresponding to generator $\Act$ in \eqref{eq:Act-example}). As the origin is exit and $L$ is regular, the short-rate $X$ can hit both endpoints $0$ and $L$.  However, the trajectories spend more time near the origin than they do near $L$ before reaching either due to the fact that $\sig(x) = a \sqrt{x} \to 0$ as $x \to 0$.

\begin{proposition}
\label{thm:k12-eigenfunctions}
For $k=1/2$ (dynamics of $X$ given by \eqref{eq:k12-SDE}) the origin is exit, the spectrum $\sig(\Act) = \sig_d(\Act)$ is purely discrete and the eigenfunctions $\psi_n(x) := \psi_{\lam_n}(x)$ \eqref{eq:psi-ode}, normalized as \eqref{eq:normalization}, have the following form 
\begin{align}
\psi_n(x)
	&= \sqrt{\frac{a}{2c_n}} x M\left(1 - \tfrac{\lambda_n}{a\sqrt2}, 2; -\tfrac{2\sqrt2}{a} x\right), \label{eq:k12-eigenfunctions} \\
c_n 
	&:= \int_{0}^{a}{x\ee^{\frac{2\sqrt2 x}{a}} \left(M\left(1 - \tfrac{\lam_n}{a\sqrt2}, 2; -\tfrac{2\sqrt2}{a} x\right)\right)^2 \dd x}, \label{eq:cn}
\end{align} 
where $M(\alpha, \beta; z)$ is a Kummer's confluent hypergeometric function, that is \cite[Chapter 13]{NIST:DLMF}
\begin{align}
M(\alpha, \beta ; z) 
	&:= \sum_{k=0}^{\infty}{\frac{\alpha^{(k)} z^k}{\beta^{(k)} k!}}, & 
\alpha^{(k)} 
	&:= \alpha(\alpha+1)\dots (\alpha+k-1), &
\alpha^{(0)}
	&:= 1, \label{eq:hypergeometric}
\end{align}
while the eigenvalues $\lam_n$ are solutions of  
\begin{align}
M\left(1-\tfrac{\lam_n}{a\sqrt2}, 2; -\tfrac{2\sqrt2}{a}L\right) 
	&= 0.\label{eq:k12-eigeneq}
\end{align}
Then, the transition density \eqref{eq:Gamma-pde} admits the representation
\begin{align}
\Gamt(t, x; T, y) 
	&= \ee^{\frac{2\sqrt2 y}{a}} \sum_{n=1}^{\infty}{\frac{\ee^{\lambda_n (T- t) }}{c_n} x M\left(1 - \tfrac{\lam_n}{a\sqrt2}, 2; -\tfrac{2\sqrt2}{a} x\right) M\left(1 - \tfrac{\lam_n}{a\sqrt2}, 2; -\tfrac{2\sqrt2}{a} y\right)   } . \label{eq:k12-Gamma}
\end{align}
\end{proposition}

\begin{proof}
See Appendix \ref{sec:k12-eigenfunctions}
\end{proof}

\noindent 
In Figure \ref{fig:eigenfunctions-12}, with $L=1, a=1$, we plot the first four normalized eigenfunctions \eqref{eq:k12-eigenfunctions}, corresponding to $\lam_1 = -2.16096, \lam_2 = -6.48742, \lam_3=-13.2721, \lam_4=-22.5243$. Smaller $\lam_n$ is associated with a more oscillatory eigenfunction.
\\[0.5em]
In Figure \ref{fig:gamma}, we plot the transition density $\Gamt(t,x;T,y)$ \eqref{eq:k12-Gamma} as a function of $y$ with $x = 0.5$ and $t = 0$ fixed, and $T$ going from $0.05$ to $0.3$ in steps of $0.05$. The area under the density curve $\int_{0}^{1}{\Gamt(x, t; T, y) \dd y}$ is less than $1$ for $T > 0$ because $X_T$ has a positive probability of being located at one of the endpoints.

\begin{remark}
\label{rem:k12-density}
The transition density $\Gamt$, which is needed to write the explicit solution \eqref{eq:u-explicit} of the pricing PDE \eqref{eq:u-pde}, is for the short-rate $X$ under the equivalent probability measure $\Pbt$ \eqref{eq:radon-nikodym}. Under $\Pbt$, from the expression \eqref{eq:Act-example} when $k=1/2$, $X$ has the following SDE
\begin{align}
\dd X_t 
	&= \Ib{X_t \in (0, L)}\Big( a\sqrt{2} X_t \dd t + a\sqrt{X_t} \dd W_t \Big) . \label{eq:k12-equivprob-SDE}
\end{align}
\end{remark}

\noindent 
To compute bond prices (and the prices of other financial derivatives), we must compute an expression for $f(x)$ \eqref{eq:f-explicit}.  From \eqref{eq:model-phi} we have
\begin{align}
f(x)
	&=	\begin{cases}-\frac{\sqrt{8 x^{2k+1}}}{a(2k+1)}, &k \neq -\frac{1}{2}, \\
	-\frac{\sqrt{2}}{a} \log x, &k = -\frac{1}{2}.
	\end{cases}  & \label{eq:f-function}
\end{align}
\begin{proposition}
For $k=1/2$ short-rate model \eqref{eq:k12-SDE} the bond price \eqref{eq:bond-yield} has the following analytic expression 
\begin{align}
B_t^T
	&= \ee^{-x (T-t)} + x \ee^{\frac{\sqrt2}{a} x} \sum_{n=0}^{\infty}{\frac{M\left(1 - \tfrac{\lam_n}{a\sqrt2}, 2; -\tfrac{2\sqrt2}{a} x\right)}{c_n}\int_{0}^{L}{ h(t, T; \xi, \lam_n) M\left(1 - \tfrac{\lam_n}{a\sqrt2}, 2; -\tfrac{2\sqrt2}{a} \xi\right) \dd \xi}  }, \label{eq:k12-bond} \\ 
h(t, T; \xi, \lam ) 
	&:= \frac{a^2 \xi  \ee^{\xi  \left(\frac{\sqrt{2}}{a}-(T-t)\right)} \left(-((\lam +\xi ) (t-T) ((\lam +\xi ) (t-T)-2))+2
	\ee^{(\lam +\xi ) (T-t)}-2\right)}{2 (\lam +\xi )^3},
\end{align}
where $c_n$ is defined as in \eqref{eq:cn}, $M(\alpha, \beta; z)$ is the Kummer confluent hypergeometric function \eqref{eq:hypergeometric}, and $\lam_n$ are solutions of \eqref{eq:k12-eigeneq}.
\end{proposition}
\begin{proof}
By definition \eqref{eq:bond-yield}, we let $g(x) = 1$ in \eqref{eq:u-explicit}. When $k=1/2$, the expressions for \eqref{eq:Act-example} and \eqref{eq:f-function} are 
\begin{align}
\Act
	&= a \sqrt{2} x \partial_x + \tfrac{1}{2} a^2 x \partial_x^2, &
f(x) 
	&= - \frac{\sqrt{2}}{a} x. \label{eq:k12-A-f}
\end{align}
Plugging \eqref{eq:k12-A-f} into \eqref{eq:q} results in 
\begin{align}
q(t, x) 
	&= \frac12  a^2 (T-t)^2 x \ee^{-x(T-t + \frac{\sqrt2}{a}) } .\label{eq:k12-q}
\end{align}
Finally, substituting \eqref{eq:k12-Gamma} and \eqref{eq:k12-q} into \eqref{eq:u-explicit} we get 
\begin{align} 
u(t, x)
	&= \ee^{-x(T-t)} \\
	&+\frac{a^2}{2}x \ee^{\frac{\sqrt{2} x}{a}}\sum_{n=0}^{\infty}{\frac{M\left(1 - \tfrac{\lam_n}{a\sqrt2}, 2; -\tfrac{2\sqrt2}{a} x\right)}{c_n} \int_{t}^{T}{\dd s \int_{0}^{L}{\dd \xi \ee^{\frac{\sqrt2 \xi}{a}+ \lam_n (s- t) -\xi(T-s)} M\left(1 - \tfrac{\lam_n}{a\sqrt2}, 2; -\tfrac{2\sqrt2}{a} \xi\right)   } }(T-s)^2 \xi  } .
\end{align}
Integrating with respect to $s$ first results in \eqref{eq:k12-bond}. 
\end{proof}

\noindent 
In Figures \ref{fig:bond} and \ref{fig:yield} we fix parameters $L=1,a=1$ and plot both the bond price $B_0^T$ \eqref{eq:k12-bond} and the corresponding yield curve $Y_0^T$ \eqref{eq:bond-yield} with $X_0 = x \in \{\frac{1}{3}, \frac{1}{2}, \frac{2}{3}\}$. Notice that the yield curve has an inverted pattern.

\subsection{Example: $k = -1/2$, the origin is natural}
\label{subsec:kneg12}
For $k = -1/2$, the SDE \eqref{eq:our-short-rate} of the short-rate $X$ becomes 
\begin{align}
\dd X_t
	&= \Ib{X_t \in (0, L)} \left( \tfrac{a^2}{2} X_t^2 \dd t + a X_t^{3/2} \dd W_t \right) . \label{eq:kneg12-SDE}
\end{align} 
In Figure \ref{fig:kneg12-trajectories}, with $a=L=1$, we plot sample trajectories of the short-rate \eqref{eq:kneg12-SDE} under $\Pb$ (corresponding to generator $\Ac$  in \eqref{eq:A}) as well as the corresponding trajectories for the short-rate $X$ under $\Pbt$ (corresponding to generator $\Act$ in \eqref{eq:Act-example}). Note that the short-rate $X_t$ can hit the right endpoint $L$, which is regular, while it cannot reach the origin, which is natural. 

\begin{proposition}
\label{thm:kneg12-eigenfunctions}
For $k=-1/2$ model \eqref{eq:kneg12-SDE} the origin is natural, the spectrum of $\Act$ \eqref{eq:Act-example} is purely continuous in $(-\infty, 0)$ and the improper eigenfunctions $\psi(\rho, x) := \psi_{\lam(\rho)}(x)$ \eqref{eq:psi-ode} normalized as \eqref{eq:normalization-continuous} have the following form 
\begin{align}
\psi(\rho, x) 
	&:= \frac{\sqrt{\frac{L\rho}{2}}\pi \csc\left(\tfrac{2\sqrt2 \pi}{a}\right)}{\left|K_{\frac{2\sqrt2}{a}}\left(\ii \tfrac{2\sqrt2 \rho}{a}\right)\right|x^{\frac{\sqrt2}{a}}} \left(J_{\frac{2\sqrt2}{a}}\left(\tfrac{2\sqrt2 \rho}{a}\right)J_{-\frac{2\sqrt2}{a}}\left(\tfrac{2\sqrt2 \rho}{a} \sqrt{\tfrac{L}{x}}\right) - J_{-\frac{2\sqrt2}{a}}\left(\tfrac{2\sqrt2 \rho}{a}\right)J_{\frac{2\sqrt2}{a}}\left(\tfrac{2\sqrt2 \rho}{a} \sqrt{\tfrac{L}{x}}\right)\right),\label{eq:kneg12-eigenfunctions}
\end{align}
where $J_\alpha(x)$ is the Bessel function of the first kind and $K_\alpha(x)$ is the modified Bessel function of the second kind. Then, the transition density \eqref{eq:Gamma-pde} admits the representation
\begin{align}
\Gamt(t,x;T,y)
	&= \frac{2}{a^2} y^{-2+\frac{2\sqrt2}{a}} \int_{0}^{\infty}{\ee^{-L\rho^2 (T-t) } \psi(\rho, x)\psi(\rho, y) \dd \rho }, 
	\label{eq:kneg12-Gamma}
\end{align}
\end{proposition}

\begin{proof}
See Appendix \ref{sec:kneg12-eigenfunctions}.
\end{proof}
\begin{remark}
Similar to Remark \ref{rem:k12-density}, we note that the transition density $\Gamt$ in \eqref{eq:kneg12-Gamma} is that of the short-rate $X$ \eqref{eq:kneg12-SDE}, but under the equivalent probability measure $\Pbt$ \eqref{eq:radon-nikodym}. Under $\Pbt$, from the expression \eqref{eq:Act-example}, when $k=-1/2$, the short-rate $X$ has the following dynamics
\begin{align}
\dd X_t 
	&= \Ib{X_t \in (0, L)}\Big( (\tfrac{1}{2}a^2+ a\sqrt{2})X_t^2\dd t + aX_t^{3/2}\dd W_t \Big) . \label{eq:kneg12-equivprob-SDE}
\end{align}
\end{remark}

\begin{remark}
From the proof Proposition \ref{thm:kneg12-eigenfunctions} together with \eqref{eq:normalization-continuous} we can derive that for $\gamma > 0$, $\rho_0 > 0$
\begin{align}
\delta(\rho - \rho_0)
	&= \frac{L\pi^2 \csc^2(\gamma \pi) \sqrt{\rho \rho_0}}{2|K_\gamma(\ii \gamma \rho) K_\gamma(\ii \gamma \rho_0)|} \\ 
	&\cdot \int_{0}^{L}{\frac{1}{x^\gamma} \left(J_\gamma (\gamma \rho)J_{-\gamma}(\gamma \rho \sqrt{\tfrac{L}{x}}) - J_{-\gamma}(\gamma \rho) J_\gamma(\gamma \rho \sqrt{\tfrac{L}{x}})\right)\left(J_\gamma (\gamma \rho_0)J_{-\gamma}(\gamma \rho_0 \sqrt{\tfrac{L}{x}}) - J_{-\gamma}(\gamma \rho_0) J_\gamma(\gamma \rho_0 \sqrt{\tfrac{L}{x}})\right) \dd x}, 
\end{align}
where we have let $\gamma := \frac{2\sqrt2}{a}$ in \eqref{eq:kneg12-eigenfunctions}. To our knowledge, this representation for the Dirac delta function is not present in contemporary literature.
\end{remark}

\noindent 
In Figure \ref{fig:kneg12-eigenfunctions}, with $a=l=1$ we plot four improper normalized eigenfunctions \eqref{eq:kneg12-eigenfunctions}, where we use the natural log scale for the $x$ axis. Notice that as $x \to 0$ eigenfunctions exhibit highly oscillatory behavior.
In Figure \ref{fig:gamma2} we plot the transition density $\Gam(t,x;T,y)$ \eqref{eq:kneg12-Gamma} as a function of $y$ with $x = 0.5$ and $t = 0$ fixed, and $T$ ranging from $0.05$ to $0.3$ with step $0.05$. The transition density goes to $0$ in the neighborhood of the origin, as expected due to the origin being natural. The area under the density curve $\int_{0}^{1}{\Gamt(x, t, T, y) \dd y}$ is less than $1$ for $T > 0$, because there is positive probability that $X_T = L$, which is regular endpoint.

\begin{proposition}
For $k = -1/2$ corresponding to short-rate model \eqref{eq:kneg12-SDE}, the bond price \eqref{eq:bond-yield} has the following analytic expression 
\begin{align}
B_t^T
	&= \ee^{-x(T-t)} + \frac{L}{2} \pi^2 \csc^2\left(\tfrac{2\sqrt2 \pi}{a}\right) \int_{0}^{L}\dd \xi \int_{0}^{\infty}\dd\rho \, {h(t, T, \xi, \rho) \theta(\rho, x) \theta(\rho, \xi) \rho }, \label{eq:kneg12-bond} \\
\theta(\rho, x)
	&:= \left|K_{\frac{2\sqrt2}{a}}\left(\ii \tfrac{2\sqrt2 \rho}{a}\right)\right|^{-1} \left(J_{\frac{2\sqrt2}{a}}\left(\tfrac{2\sqrt2 \rho}{a}\right)J_{-\frac{2\sqrt2}{a}}\left(\tfrac{2\sqrt2 \rho}{a} \sqrt{\tfrac{L}{x}}\right) - J_{-\frac{2\sqrt2}{a}}\left(\tfrac{2\sqrt2 \rho}{a}\right)J_{\frac{2\sqrt2}{a}}\left(\tfrac{2\sqrt2 \rho}{a} \sqrt{\tfrac{L}{x}}\right)\right), \\
h(t, T, \xi, \rho)
	&:= \tfrac{\left(L \rho ^2+\xi \right) \ee^{L \rho ^2 (t-T)}-\ee^{\xi  (t-T)} \left(L^2 \rho ^4 (t-T) (\xi  t-\xi  T+1)+L \rho ^2 \left(1-2 \xi ^2 (t-T)^2\right)+\xi +\xi ^3
	(t-T)^2+\xi ^2 (T-t)\right)}{\left(\xi -L \rho ^2\right)^3}, 
\end{align}
where $J_\alpha(x)$ is the Bessel function of the first kind and $K_\alpha(x)$ is the modified Bessel function of the second kind.
\end{proposition}
\begin{proof}
By definition \eqref{eq:bond-yield}, we let $g(x) = 1$ in \eqref{eq:u-explicit}. When $k=-1/2$, the expressions for \eqref{eq:Act-example} and \eqref{eq:f-function} become
\begin{align}
\Act 
	&= \Big(\frac{a^2}{2} + a \sqrt{2} \Big)x^2 \partial_x  + \frac{a^2}{2} x^3 \partial_x^2, &
f(x)
	&= -\frac{\sqrt2}{a} \log x. \label{eq:kneg12-A-f}
\end{align}
Plugging \eqref{eq:kneg12-A-f} into \eqref{eq:q} results in 
\begin{align}
q(t, x) 
	&= \frac{a^2}{2} \ee^{-(T-t)x} x^{2-\frac{\sqrt2}{a}}(t-T)(1+tx-Tx). \label{eq:k-neg12-q}
\end{align}
Finally, substituting \eqref{eq:kneg12-Gamma} and \eqref{eq:k-neg12-q} into \eqref{eq:u-explicit} we obtain 
\begin{align}
u(t, x) 
	&= \ee^{-x(T-t)} + \int_t^T \dd s \int_0^L \dd \xi \int_{0}^{\infty} \dd \rho \, \ee^{-L \rho^2(s-t)-(T-s)\xi} (s-T)(1+s\xi - T\xi)x^{\frac{\sqrt2}{a}}\psi(\rho, x)\xi^{\frac{\sqrt2}{a}}\psi(\rho, \xi) . \label{eq:kneg12-u-explicit}
\end{align}
Integrating with respect to $s$ first in \eqref{eq:kneg12-u-explicit} and simplifying the expression results in \eqref{eq:kneg12-bond}.
\end{proof}

\noindent 
In figures \ref{fig:bond2} and \ref{fig:yield2}, with $a=L=1$ fixed, we plot both the bond price $B_0^T$ \eqref{eq:kneg12-bond} and the corresponding yield curve $Y_0^T$ \eqref{eq:bond-yield} where $X_0 = x \in \{\frac{1}{3}, \frac{1}{2}, \frac{2}{3}\}$.  For a higher value of $x$ the yield curve has a humped shape at the beginning, before decreasing, while for a lower value of $x$ it appears to be more flat.


\section*{Acknowledgments}
Aram Karakhanyan acknowledges partial support from EPSRC grant EP/S03157X/1 Mean curvature measure of free boundary.

%
%

\appendix

%
%

\section{Proof of Proposition \ref{thm:origin}}
\label{sec:origin}
Consider the \textit{scale measure}
\begin{align}
\St[x, y] 
	&:= \int_{x}^{y}{\st (z) \dd z}, \quad x, y \in (0, L), \label{eq:scale-measure}
\end{align}
where $\st (x)$ is a scale density \eqref{eq:st}, together with the limiting case at the origin
\begin{align}
\St(0, z] 
	&= \lim_{x \to 0+} S[x, z] \label{eq:scale-boundary}, \quad z \in (0, L). 
\end{align}
The boundary classification of the origin depends on the following integrals
\begin{align}
I_0 
	&= \int_{0}^{\varepsilon}{\St(0, z]\mt (z)\dd z}, \label{eq:I-integral} \\
J_0 
	&= \int_{0}^{\varepsilon}{\St[z, \varepsilon]\mt (z)\dd z}, \label{eq:J-integral} 
\end{align}
where $\varepsilon \in (0,L)$ and $\mt(x)$ is a speed density \eqref{eq:mt}. In particular, according to \cite{linetskybook}, the origin is 
\begin{enumerate}
	\item \textit{regular} if $I_0 < \infty$ and $J_0 < \infty$,
	\item \textit{exit} if $I_0 < \infty$ and $J_0 = \infty$,
	\item \textit{natural} if $I_0 = \infty$ and $J_0 = \infty$.
\end{enumerate}
Thus, in order to classify the origin, it is sufficient to determine whether each of the integrals $I_0$ and $J_0$ converges or diverges. 
\\[0.5em]
For simplicity, denote 
\begin{align}
w(x) 
	&= \frac{\sqrt2}{a(k+\frac12)}x^{k+\frac12},
\end{align}
then evaluating \eqref{eq:scale-measure} we get the expression for the scale measure 
\begin{align}
\St[x, y] 
	&= \tfrac{Ca}{2\sqrt2} \begin{cases}
		\exp(-2w(x)) - \exp(-2w(y)), &k \neq -\frac{1}{2}, \\
		x^{-\frac{2\sqrt2}{a}} - y^{-\frac{2\sqrt2}{a}}, & k = -\frac{1}{2}.
	\end{cases}\label{eq:scale-measure-explicit}
\end{align}
Next, we compute $\St(0, \varepsilon]$ as in \eqref{eq:scale-boundary}
\begin{align}
\St(0, z] 
	&= \begin{cases}
		\frac{Ca}{2\sqrt2} \left(1 - \exp\left(-2w(z)\right) \right), &k > -\frac12, \\
		\infty, &k \leq -\frac12.
	\end{cases}\label{eq:scale-measure-limit}
\end{align}
Then by plugging \eqref{eq:mt} and \eqref{eq:scale-measure-limit} into \eqref{eq:I-integral} for $k > -\frac{1}{2}$ we get 
\begin{align}
I_0 
	&=  \tfrac{1}{a\sqrt2} \int_{0}^{\varepsilon}{z^{k-\frac32}\left(\exp\left(\tfrac{2\sqrt2}{a(k+\frac12)} z^{k+\frac12}\right) - 1\right) \dd z}.
\end{align}
Using the asymptotic equivalence 
\begin{align}
\exp\left(\tfrac{2\sqrt2}{a(k+\frac12)} z^{k+\frac12}\right) - 1 
	&\sim \tfrac{2\sqrt2}{a(k+\frac12)} z^{k+\frac12} \quad \text{ as }z \to 0, 
\end{align}
it follows that convergence of $I_0$ coincides with convergence of 
\begin{align}
	\tfrac{2}{a^2(k+\frac12)}\int_{0}^{\varepsilon}{z^{2k-1} \dd z}
\end{align}
but the latter integral converges for $k > 0$ and diverges for $k \in (-\frac12, 0]$. Together with \eqref{eq:scale-measure-limit} for $k \leq -\frac{1}{2}$ this results in 
\begin{align}\label{eq:cases-I0}
	\begin{cases}
		I_0 <\infty, &k > 0, \\
		I_0 = \infty, &k \leq 0.
	\end{cases}
\end{align}
Plugging \eqref{eq:mt} and \eqref{eq:scale-measure-explicit} into \eqref{eq:J-integral} for $k\neq -\frac12$
\begin{align}
J_0 
	&= \tfrac{1}{a\sqrt2} \int_{0}^{\varepsilon}{\left(1 - \exp\left(\tfrac{2\sqrt2}{a(k+\frac12)} \left(z^{k+\frac12} - \varepsilon^{k+\frac12}\right)\right)\right)z^{k-\frac32} \dd z}
\end{align}
Notice that since $z^{k+\frac12} \to 0$ as $z \to 0$ for $k > -\frac{1}{2}$, and $z^{k+\frac12} \to \infty$ for $k < -\frac{1}{2}$, so we have asymptotic equivalence 
\begin{align}
\left(1 - \exp\left(\tfrac{2\sqrt2}{a(k+\frac12)} \left(z^{k+\frac12} - \varepsilon^{k+\frac12}\right)\right)\right)z^{k-\frac32} 
	\sim D z^{k-\frac32} \quad \text{ as } z \to 0,
\end{align} 
where 
\begin{align}
D 
	&= \begin{cases}
	\left(1 - \exp\left(-\tfrac{2\sqrt2}{a(k+\frac12)}\varepsilon^{k+\frac12}\right)\right), &k > -\frac{1}{2}, \\
	1, &k < -\frac{1}{2}.
\end{cases}.
\end{align}
Thus convergence of $J_0$ \eqref{eq:J-integral} coincides with convergence of 
\begin{align}
\frac{D}{a\sqrt2} \int_{0}^{\varepsilon}{z^{k-\frac32} \dd z},
\end{align}
which converges for $k > -\frac12$, so $J_0 < \infty$ in this case, and diverges for $k < -\frac12$, so $J_0 = \infty$. For $k = -\frac12$  
\begin{align}
J_0 
	&= \tfrac{1}{a \sqrt2}\int_{0}^{\varepsilon}{\left(z^{-\frac{2\sqrt2}{a}} - \varepsilon^{-\frac{2\sqrt2}{a}}\right)z^{-2+\frac{2\sqrt2}{a}} \dd z} = \tfrac{1}{a\sqrt2} \int_{0}^{\varepsilon}{\left(z^{-2} - \varepsilon^{-\frac{2\sqrt{2}}{a}} z^{-2+\frac{2\sqrt2}{a}} \right)\dd z} =\infty, 
\end{align}
from asymptotic equivalence 
\begin{align}
z^{-2} -  \varepsilon^{-\frac{2\sqrt{2}}{a}} z^{-2+\frac{2\sqrt2}{a}} 
	&= z^{-2}\left(1-\varepsilon^{-\frac{2\sqrt{2}}{a}} z^{\frac{2\sqrt2}{a}}\right)  \sim z^{-2}, \quad \text{ as } z \to 0. 
\end{align}
Thus, we get 
\begin{align}
	\begin{cases}
		J_0 < \infty, &k > \frac{1}{2}, \\
		J_0 = \infty, &k \leq \frac{1}{2}.
	\end{cases}
\end{align}
This, combined with \eqref{eq:cases-I0} proves the proposition. 

%
%

\section{Proof of Proposition \ref{thm:k12-eigenfunctions}}
\label{sec:k12-eigenfunctions}

From Proposition \ref{thm:origin} the origin is an exit boundary, and 
from Proposition \ref{thm:spectrum} the spectrum is purely discrete. Plugging $k=1/2$ into \eqref{eq:mu-sig-example}, \eqref{eq:st}, \eqref{eq:mt} with $C=1$, and \eqref{eq:Liouville-transformation} we get 
\begin{align}
\sig(x)
	&= a\sqrt{x}, & 
\st(x)
	&= \ee^{-\frac{2\sqrt2}{a} x}, & 
\mt(x)
	&= \frac{2\ee^{\frac{2\sqrt2}{a} x}}{a^2 x}, &
p(x) 
	&= \frac{2\sqrt{x}}{a}. \label{eq:k12-sigma-st-mt-p}
\end{align}
Thus, from \eqref{eq:eta-ode-2} and \eqref{eq:k12-sigma-st-mt-p} we derive the following equation 
\begin{align}
\eta_\lam''(z)
	&=	\frac1{4z^2}(8 \lam  z^2 + 2a^2 z^{4} +3)\eta_\lam(z), &
\eta_\lam(0) 
	&= 0, &
\eta_\lam\left(\tfrac{2\sqrt{L}}{a}\right) 
	&= 0. 
\end{align}
Applying the first BC to the analytic solution \eqref{eq:eta-k12-solution} we get that $c_1 = 0$, and expanding $L^{(-1)}_{-\frac{\lam}{\sqrt{2a}}}\left(\tfrac{az^2}{2}\right)$, the expression for the eigenfunctions becomes 
\begin{align}
\eta_\lam(z) 
	&= c z^{\frac32} \ee^{-\frac{az^2}{2\sqrt2}} M\left(1+\tfrac{\lam}{a\sqrt2},2;\tfrac{a z^2}{\sqrt2}\right) ,
	\label{eq:potential-eigenfunctions}
\end{align}
where $c$ is a constant coefficient. Before applying the second BC, we return back to the original variable $x$ by plugging \eqref{eq:k12-sigma-st-mt-p} and \eqref{eq:potential-eigenfunctions}, into \eqref{eq:z-def}, resulting in 
\begin{align}
\psi_\lam(x) 
	&= c \ee^{-\frac{2\sqrt2}{a}x} x M\left(1 + \tfrac{\lam}{a\sqrt2}, 2; \tfrac{2\sqrt2}{a} x\right),\label{eq:k12-psi-temp}
\end{align}
and using the Kummer's transformation
\begin{align}
	M(\alpha,\beta; z) = \ee^{z} M(\beta-\alpha,\beta; -z),
\end{align}
\eqref{eq:k12-psi-temp} becomes 
\begin{align}
\psi_\lam (x) 
	&= c x M\left(1 - \tfrac{\lam}{a\sqrt2}, 2; -\tfrac{2\sqrt2}{a} x\right).
\end{align}
Now, applying the second BC \eqref{eq:psi-bc-L}
\begin{align}
\psi_\lam (L) 
	&= c L \cdot M\left(1-\tfrac{\lam}{a\sqrt2}, 2; -\tfrac{2\sqrt2}{a}L\right) = 0.
\end{align}
Thus, the eigenvalues $\lam_n$ are solutions of this equation, or of \eqref{eq:k12-eigeneq} equivalently. The constant $c$ for each $\psi_{n} = \psi_{\lam_n}$ is determined from normalization \eqref{eq:normalization} with $\mt$ \eqref{eq:k12-sigma-st-mt-p}.
\begin{align}
\|\psi_n\|^2_{\mt} 
	&= \langle \psi_n, \psi_n \rangle_{\mt} = \frac{2 c^2 }{a^2}\int_{0}^{L}{\ee^{\frac{2\sqrt2}{a} x}x\left(M\left(1-\tfrac{\lam_n}{a\sqrt2}, 2; -\tfrac{2\sqrt2}{a}L\right)\right)^2 \dd x} = 1,
\end{align}
so if we introduce $c_n$ \eqref{eq:cn}, then 
\begin{align}
c 
	&= \sqrt{\frac{a}{2c_n}},
\end{align}
which produces \eqref{eq:k12-eigenfunctions}. Finally, since the spectrum is purely discrete, i.e. $\sig(\Act) = \sig_{p}(\Act)$ and $\sig_c(\Act) = 0$, the expression for the transition density $\Gamt(t, x; T, y)$ \eqref{eq:k12-Gamma} follows from the general spectral representation \eqref{eq:Gamma}.

%
%

\section{Spectral representation for the natural boundary case}
\label{sec:greens}

In practice, when the origin is natural and, consequently, the continuous spectrum is nonempty (see Proposition \ref{thm:spectrum}), the eigenfunction representation is written by finding the Green's function first \cite[Chapter 2]{borodin2002handbook}
\begin{align}
G_\lam (x, y) 
	&= \frac{\psi_\lam(\min(x,y)) \phi_\lam(\max(x,y))}{w_\lam}, & 
w_\lam 
	&:= \frac{1}{\st(x)}\left( \psi_\lam' (x) \phi_\lam (x) - \psi_\lam(x) \phi_\lam'(x) \right), \label{eq:greensf-wronskian}
\end{align}
where $\psi_\lam(x)$ and $\phi_\lam(x)$ are solutions of \eqref{eq:psi-ode} s.t. $\psi_\lam$ is increasing, while $\phi_\lam$ is decreasing and satisfies \eqref{eq:psi-bc-L}. Both functions should satisfy some additional conditions at the origin to ensure Green's function uniqueness, that is 
\begin{align}
\psi_\lam(0+) 
	&:= \lim_{x \to 0+}{\psi_\lam(x)} = 0, &
\phi_\lam(0+)
	&:= \lim_{x \to 0+}{\phi_\lam(x)} = +\infty, \\
\psi_\lam^+(0+) 
	&:= \lim_{h \to 0+}{\frac{\psi_\lam(0+h) - \psi_\lam(0+)}{S(0, h]}} = 0, &
\phi_\lam^+(0+)
	&= \lim_{h \to 0+}{\frac{\phi_\lam(0+h) - \phi_\lam(0+)}{S(0, h]}} = -\infty, 
\end{align}
where $S(0, h]$ is defined as in \eqref{eq:scale-boundary}. These conditions will typically be satisfied for $\psi_\lam$ and $\phi_\lam$ under the prior assumptions, but when they are not, one must change $\psi_\lam$ and $\phi_\lam$ accordingly. Under the aforementioned conditions on $\psi_\lam(x)$ and $\phi_\lam(x)$, the Wronskian $w_\lam$ \eqref{eq:greensf-wronskian} will be independent of $x$ which justifies the notation. 
\\[0.5em]
Then, applying inverse Laplace transform to \eqref{eq:greensf-wronskian} results in 
\begin{align}
\Gamt(t, x; T, y) 
	&= \frac{\mt(y)}{2\pi \ii } \int_{c - \ii \infty}^{c + \ii \infty}{\ee^{\lam (T -t )} G_\lam(x, y) \dd \lam },\label{eq:densitylaplace}
\end{align}
which produces a discrete part of \eqref{eq:Gamma} when the Green's function has isolated singularities, and a continuous part when there is a branch cut for $\lam$. 

%
%

\section{Proof of Proposition \ref{thm:kneg12-eigenfunctions}}
\label{sec:kneg12-eigenfunctions}

From Proposition \ref{thm:origin} the origin is a natural boundary, and from Proposition \ref{thm:spectrum} the spectrum is purely continuous in $(-\infty, 0)$. Plugging $k=-1/2$ into \eqref{eq:mu-sig-example}, \eqref{eq:st}, \eqref{eq:mt} with $C=1$, and \eqref{eq:Liouville-transformation} we get 
\begin{align}
\sig(x)
	&= a x^{3/2}, & 
\st(x)
	&= x^{-1-\frac{2\sqrt2}{a}}, &
\mt(x)
	&= \frac{2}{a^2} x^{-2+\frac{2\sqrt2}{a}}, &
p(x)
	&= -\frac{2}{a \sqrt{x}}. \label{eq:kneg12-sigma-st-mt-p}
\end{align}
From \eqref{eq:eta-ode-2} we derive the following equation 
\begin{align}
\eta'' 
	&=	\frac1{4z^2}(8 \lam  z^2 + 2a^2 z^{4} +3)\eta.
\end{align}
The general solution to this ODE is given by \eqref{eq:eta-k-neg12-solution}. First, we return back to the original variable $x$ by plugging \eqref{eq:eta-k-neg12-solution} and \eqref{eq:kneg12-sigma-st-mt-p} \eqref{eq:z-def}. This results in the general solution of the eigenvalue equation \eqref{eq:psi-ode}
\begin{align}
\psi_\lam(x) 
	&= c_1 x^{-\frac{\sqrt2}{a}}J_{\frac{2\sqrt2}{a}}\left(\tfrac{2\ii \sqrt{2}}{a} \sqrt{\tfrac{\lam}{x}}\right) + c_2 x^{-\frac{\sqrt2}{a}}Y_{\frac{2\sqrt2}{a}}\left(\tfrac{2\ii \sqrt{2}}{a} \sqrt{\tfrac{\lam}{x}}\right),\label{eq:kneg12-general-psi}
\end{align}
where $J_\alpha(z), Y_\alpha(z)$ are the Bessel functions. \eqref{eq:kneg12-general-psi} is a linear combination of two independent solutions 
\begin{align}
\psi_\lam^{(1)}
	&= x^{-\frac{\sqrt2}{a}}J_{\frac{2\sqrt2}{a}}\left(\tfrac{2\ii \sqrt{2}}{a} \sqrt{\tfrac{\lam}{x}}\right), & 
\psi_\lam^{(2)} 
	&= x^{-\frac{\sqrt2}{a}}Y_{\frac{2\sqrt2}{a}}\left(\tfrac{2\ii \sqrt{2}}{a} \sqrt{\tfrac{\lam}{x}}\right).
\end{align}
Following the approach described in Appendix \ref{sec:greens}, we need to construct two solutions $\psi_\lam(x)$ and $\phi_\lam(x)$ of \eqref{eq:psi-ode} s.t. for $\lam > 0$ $\psi_\lam(x)$ is increasing, while $\phi_\lam(x)$ is decreasing and satisfies the BC \eqref{eq:psi-bc-L}
\begin{align}
\psi_\lam(x) 
	&= \frac{\pi}{2}\ee^{\frac{\left(2\sqrt2 +a\right)\pi \ii}{2a}} \psi_\lam^{(1)}(x) +  \frac{\pi}{2}\ee^{\frac{\left(2\sqrt2 +a\right)\pi \ii}{2a}}  \ii \psi_\lam^{(2)}(x) = x^{-\frac{\sqrt2}{a}} K_{\frac{2\sqrt2}{a}}\left(\tfrac{2\sqrt2}{a}\sqrt{\tfrac{\lam}{x}}\right), \label{eq:psi}\\
\phi_\lam(x)
	&= \left(K_{\frac{2\sqrt2}{a}}\left(\tfrac{2\sqrt2}{a} \sqrt{\tfrac{\lam}{L}} \right) \ee^{-\frac{\sqrt2 \pi \ii}{2a}}  - I_{\frac{2\sqrt2}{a}}\left(\tfrac{2\sqrt2}{a} \sqrt{\tfrac{\lam}{L}}\right)\frac{\pi}{2}\ee^{\frac{\left(2\sqrt2 +a\right)\pi \ii}{2a}}\right)\psi_\lam^{(1)}(x) - I_{\frac{2\sqrt2}{a}}\left(\tfrac{2\sqrt2}{a} \sqrt{\tfrac{\lam}{L}}\right)\frac{\pi}{2}\ee^{\frac{\left(2\sqrt2 +a\right)\pi \ii}{2a}} \ii \psi_\lam^{(2)}(x) \\
	&= x^{-\frac{\sqrt2}{a}}\left(I_{\frac{2\sqrt2}{a}}\left(\tfrac{2\sqrt2}{a} \sqrt{\tfrac{\lam}{x}} \right)K_{\frac{2\sqrt2}{a}}\left(\tfrac{2\sqrt2}{a} \sqrt{\tfrac{\lam}{L}} \right) - I_{\frac{2\sqrt2}{a}}\left(\tfrac{2\sqrt2}{a} \sqrt{\tfrac{\lam}{L}} \right) K_{\frac{2\sqrt2}{a}}\left(\tfrac{2\sqrt2}{a} \sqrt{\tfrac{\lam}{x}} \right) \right), \label{eq:phi}
\end{align}	
where $I_\alpha(z), K_\alpha(z)$ are the modified Bessel functions, and we have used the following connection formulas between regular and modified Bessel functions \cite[Chapter 10]{NIST:DLMF} which hold for all $z \in \Rb$
\begin{align}
I_{\alpha}(z)
	&= \ee^{-\frac{\alpha \pi \ii}{2}}J_{\alpha}(z), &
K_\alpha(z)
	&= \frac{\pi}{2} \ee^{\frac{(\alpha+1)\pi \ii}{2}} \left(J_\alpha(\ii z) + \ii Y_\alpha(\ii z)\right)
\end{align}
Plugging \eqref{eq:kneg12-sigma-st-mt-p}, \eqref{eq:psi} and \eqref{eq:phi} into the Wronskian \eqref{eq:greensf-wronskian} results in 
\begin{align}
w_\lam 
	&= \frac{1}{2}K_{\frac{2\sqrt2}{a}}\left(\tfrac{2\sqrt2}{a}\sqrt{\tfrac{\lam}{L}}\right).\label{eq:kneg12-wronskian}
\end{align}
Plugging \eqref{eq:psi}, \eqref{eq:phi} and \eqref{eq:kneg12-wronskian} into \eqref{eq:greensf-wronskian} we compute the Green's function
\begin{align}
G_\lam (x, y) 
	&= \frac{2}{x^{\frac{\sqrt2}{a}}y^{\frac{\sqrt2}{a}}K_{\frac{2\sqrt2}{a}}\left(\tfrac{2\sqrt2}{a}\sqrt{\tfrac{\lam}{L}}\right)}\\
	&\cdot \begin{cases}
		K_{\frac{2\sqrt2}{a}}\left(\tfrac{2\sqrt2}{a}\sqrt{\tfrac{\lam}{x}}\right) \left[K_{\frac{2\sqrt2}{a}}\left(\tfrac{2\sqrt2}{a}\sqrt{\tfrac{\lam}{L}}\right)I_{\frac{2\sqrt2}{a}}\left(\tfrac{2\sqrt2}{a} \sqrt{\tfrac{\lam}{y}}\right) - I_{\frac{2\sqrt2}{a}}\left(\tfrac{2\sqrt2}{a}\sqrt{\tfrac{\lam}{L}}\right)K_{\frac{2\sqrt2}{a}}\left(\tfrac{2\sqrt2}{a} \sqrt{\tfrac{\lam}{y}}\right) \right], &x < y \\
		K_{\frac{2\sqrt2}{a}}\left(\tfrac{2\sqrt2}{a}\sqrt{\tfrac{\lam}{y}}\right) \left[K_{\frac{2\sqrt2}{a}}\left(\tfrac{2\sqrt2}{a}\sqrt{\tfrac{\lam}{L}}\right)I_{\frac{2\sqrt2}{a}}\left(\tfrac{2\sqrt2}{a} \sqrt{\tfrac{\lam}{x}}\right) - I_{\frac{2\sqrt2}{a}}\left(\tfrac{2\sqrt2}{a}\sqrt{\tfrac{\lam}{L}}\right)K_{\frac{2\sqrt2}{a}}\left(\tfrac{2\sqrt2}{a} \sqrt{\tfrac{\lam}{x}}\right) \right], &y < x
	\end{cases}.
\end{align}
Finally, to get the transition density, we must invert the Laplace transform \eqref{eq:densitylaplace} with $\mt$ \eqref{eq:kneg12-sigma-st-mt-p}. Since $T -t > 0$ and $G_\lam(x,y)$ has no poles, we close the contour in the left half-plane and apply the Jordan's lemma. Then, using the substitution $\lam = -L \rho^2$ to account for the branch cut, we get 
\begin{align}
\Gamt(t,x;T,y)
	&= - \frac{L\mt(y)}{\pi \ii} \int_{0}^{\infty}{\ee^{-L\rho^2 (T-t) }\left(G_{L\rho^2 \ee^{\ii \pi}}(x,y)-G_{L\rho^2 \ee^{-\ii \pi}}(x,y) \right) \rho \dd \rho },\label{eq:Gamma-temp}
\end{align}
the integrand can be simplified as 
\begin{align}
G_{L\rho^2 \ee^{\ii \pi}}(x,y)-G_{L\rho^2 \ee^{-\ii \pi}}(x,y) 
	&= -\frac{\ii \pi^3 \csc^2\left(\tfrac{2\sqrt2 \pi}{a}\right)}{2\left|K\left(\tfrac{2\sqrt2}{a},\ii \frac{2\sqrt2 \rho}{a}\right)\right|^2 x^{\frac{\sqrt2}{a}}y^{\frac{\sqrt2}{a}}} \\
	&\cdot \left(J_{\frac{2\sqrt2}{a}}\left(\tfrac{2\sqrt2 \rho}{a}\right)J_{-\frac{2\sqrt2}{a}}\left(\tfrac{2\sqrt2 \rho}{a} \sqrt{\tfrac{L}{x}}\right) - J_{-\frac{2\sqrt2}{a}}\left(\tfrac{2\sqrt2 \rho}{a}\right)J_{\frac{2\sqrt2}{a}}\left(\tfrac{2\sqrt2 \rho}{a} \sqrt{\tfrac{L}{x}}\right)\right) \\
	&\cdot \left(J_{\frac{2\sqrt2}{a}}\left(\tfrac{2\sqrt2 \rho}{a}\right)J_{-\frac{2\sqrt2}{a}}\left(\tfrac{2\sqrt2 \rho}{a} \sqrt{\tfrac{L}{y}}\right) - J_{-\frac{2\sqrt2}{a}}\left(\tfrac{2\sqrt2 \rho}{a}\right)J_{\frac{2\sqrt2}{a}}\left(\tfrac{2\sqrt2 \rho}{a} \sqrt{\tfrac{L}{y}}\right)\right),
\end{align}
plugging this into \eqref{eq:Gamma-temp} and introducing $\psi(\rho, x)$ \eqref{eq:kneg12-eigenfunctions} results in \eqref{eq:kneg12-Gamma}. Finally, we notice that \eqref{eq:kneg12-eigenfunctions} satisfy 
\begin{align}
\left(\frac{a^2}{2}+a\sqrt2\right)x^2\partial_x \psi(\rho, x) + \frac{a^2}{2}x^3 \partial_x^2 \psi(\rho, x) = -L\rho^2 \psi(\rho, x),
\end{align}
which is \eqref{eq:psi-ode} for $\Act$ \eqref{eq:Act-example} with $k=-1/2$, and since 
\begin{align}
\left(J_{\frac{2\sqrt2}{a}}\left(\tfrac{2\sqrt2 \rho}{a}\right)J_{-\frac{2\sqrt2}{a}}\left(\tfrac{2\sqrt2 \rho}{a} \sqrt{\tfrac{L}{L}}\right) - J_{-\frac{2\sqrt2}{a}}\left(\tfrac{2\sqrt2 \rho}{a}\right)J_{\frac{2\sqrt2}{a}}\left(\tfrac{2\sqrt2 \rho}{a} \sqrt{\tfrac{L}{L}}\right)\right)
	&= 0
\end{align}
$\psi(\rho, x)$ satisfy \eqref{eq:psi-bc-L}. Therefore $\psi(\rho, x)$ are indeed improper eigenfunctions, and from the uniqueness of spectral representation \eqref{eq:Gamma} we conclude that they satisfy the normalization \eqref{eq:normalization-continuous}.

%


\bibliography{references-1.00}


\section*{Figures}

\begin{figure}[h!]
	\includegraphics[width=\textwidth]{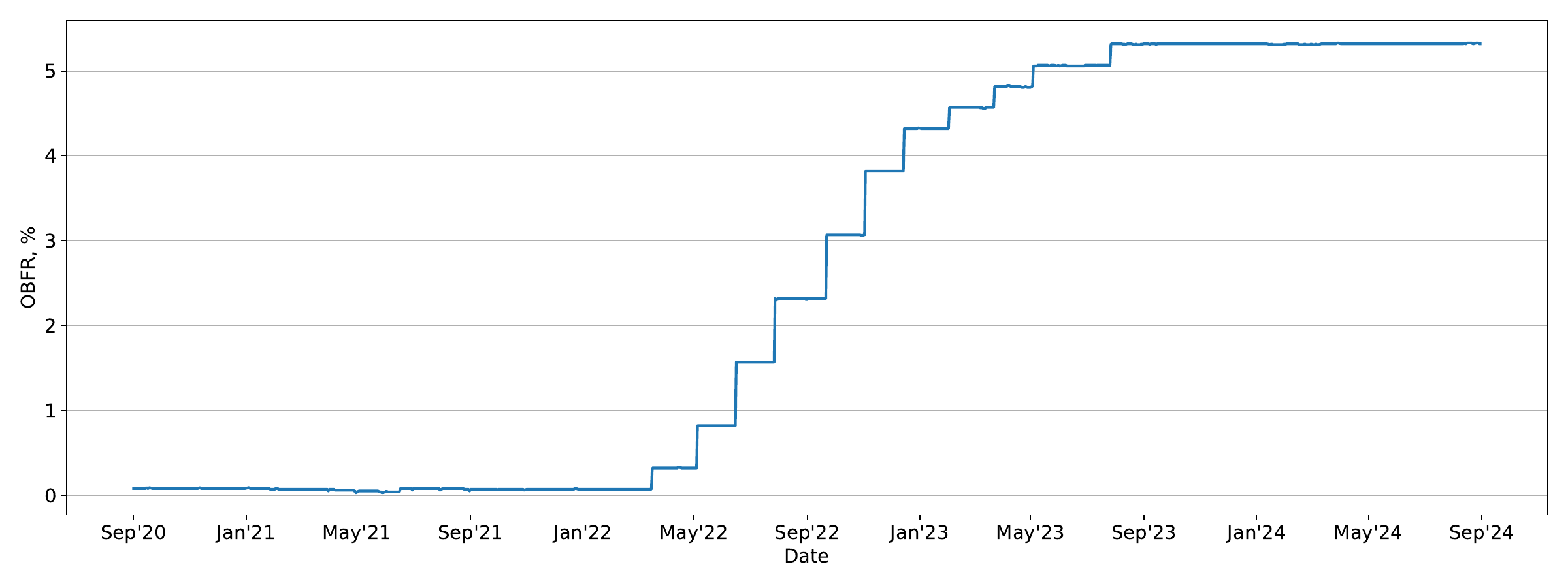}
	\centering 
	\caption{Overnight bank funding rate (OBFR). Source: Federal Reserve Bank of New York\\
	\url{https://www.newyorkfed.org/markets/reference-rates/obfr} }
	\label{fig:OBFR}
\end{figure}

\begin{figure}[h!]
	\begin{tabular}{ccc}
        \includegraphics[width=0.5\textwidth]{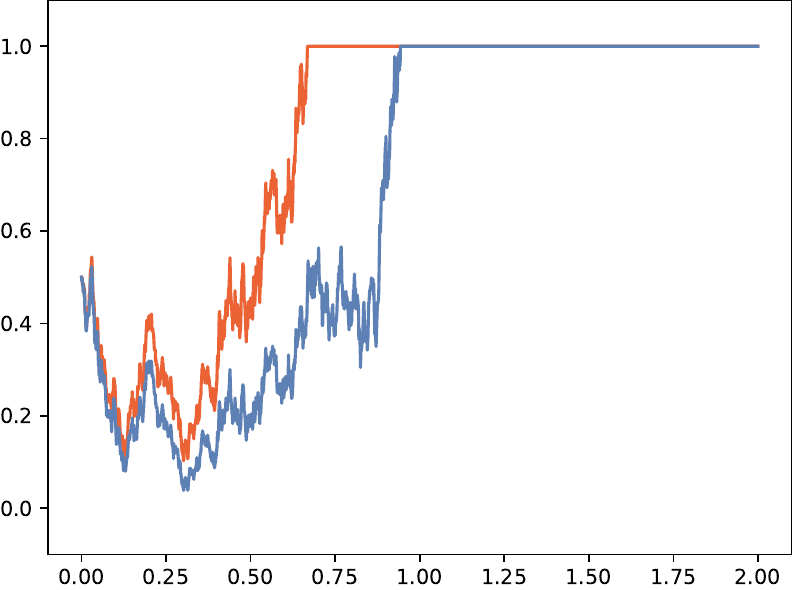} & 
        \includegraphics[width=0.5\textwidth]{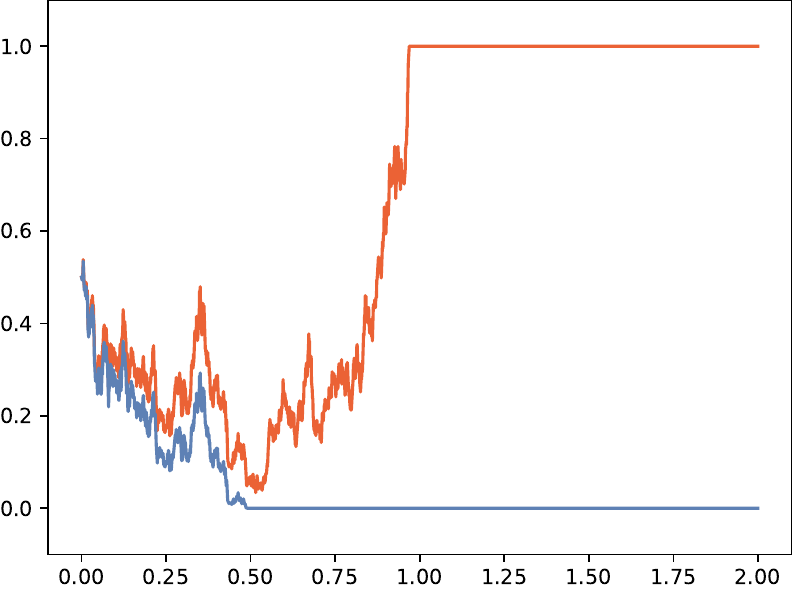} \\
        \includegraphics[width=0.5\textwidth]{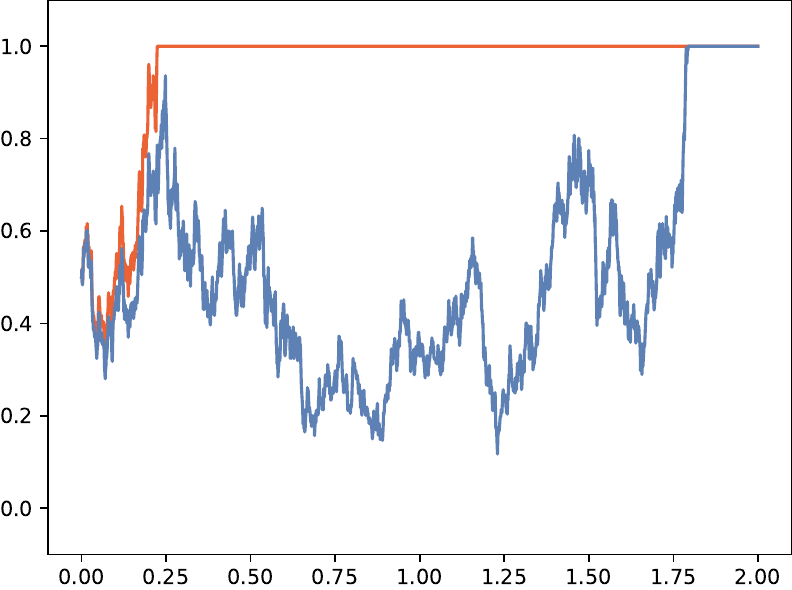} & 
        \includegraphics[width=0.5\textwidth]{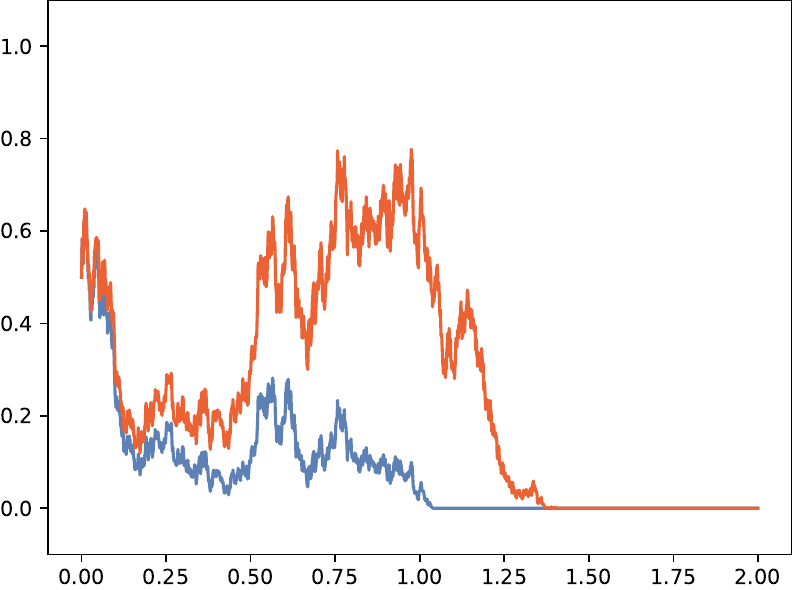} \\
    \end{tabular}
    \centering
    \caption{We plot 4 trajectories for the short-rate $X_t$ in $k=1/2$ model \ref{subsec:k12} with $L=1$, $a=1$, starting from $X_0 = 1/2$, for $t \in [0, 2]$. Blue is under the original probability measure $\Pb$ with dynamics \eqref{eq:k12-SDE}, while red is under the equivalent probability measure $\Pbt$ as introduced in Remark \ref{rem:equivprob} and dynamics \eqref{eq:k12-equivprob-SDE}.}
    \label{fig:k12-trajectories}
\end{figure}

\begin{figure}[h!]
	\includegraphics[width=0.95\textwidth]{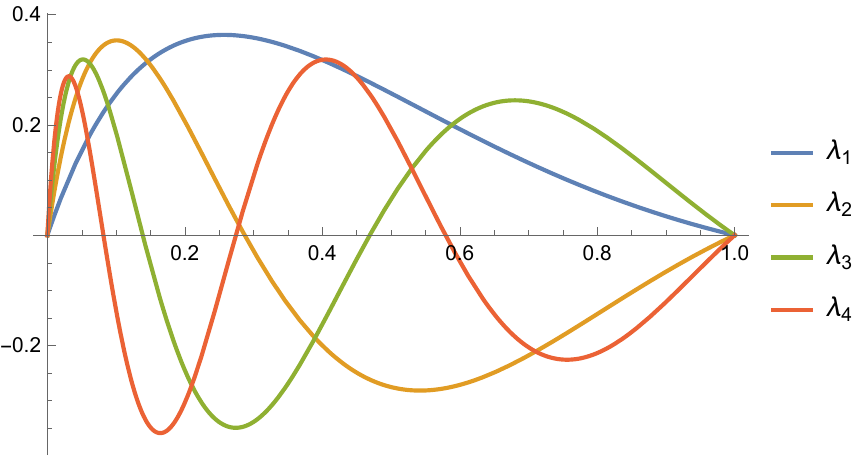}
	\centering 
	\caption{We plot the first four normalized eigenfunctions $\psi_{n}$ \eqref{eq:k12-eigenfunctions} for $L = 1$ and $a=1$, normalized as in \eqref{eq:normalization}.}
	\label{fig:eigenfunctions-12}
\end{figure}

\begin{figure}[h!]
    \centering
    \begin{tabular}{ccc}
        \includegraphics[width=0.33\textwidth]{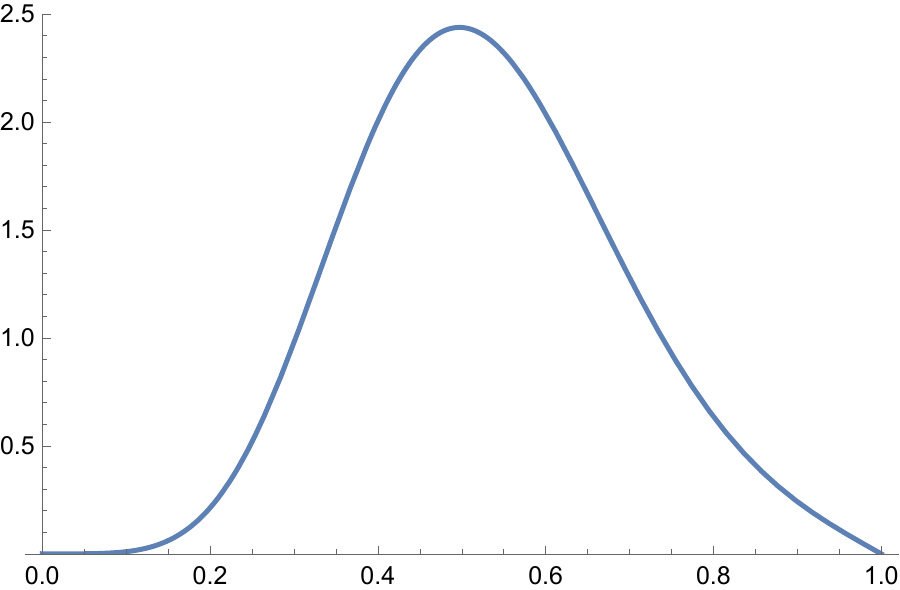} & 
        \includegraphics[width=0.33\textwidth]{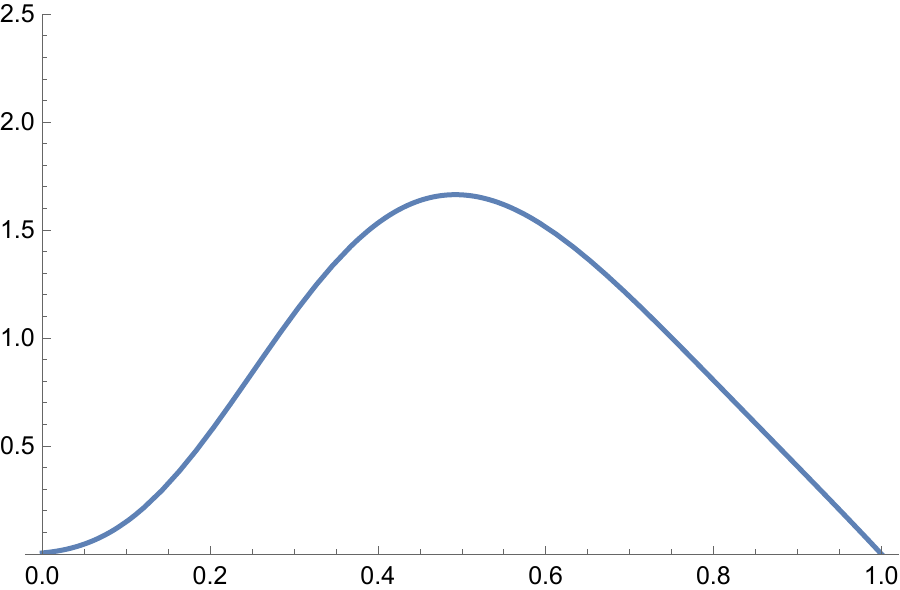} & 
        \includegraphics[width=0.33\textwidth]{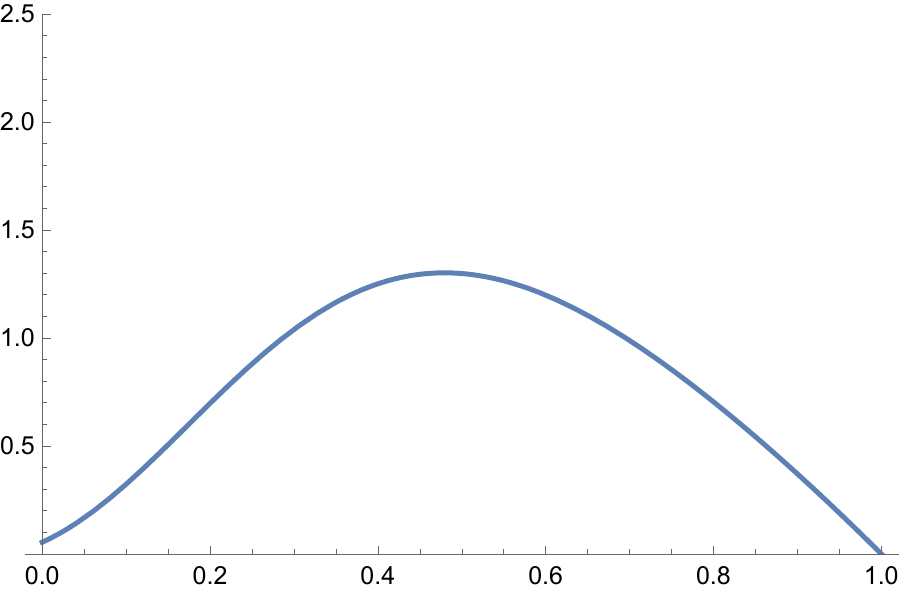} \\
		$T = 0.05$ & $T = 0.10$ & $T = 0.15$ \\ 
        \includegraphics[width=0.33\textwidth]{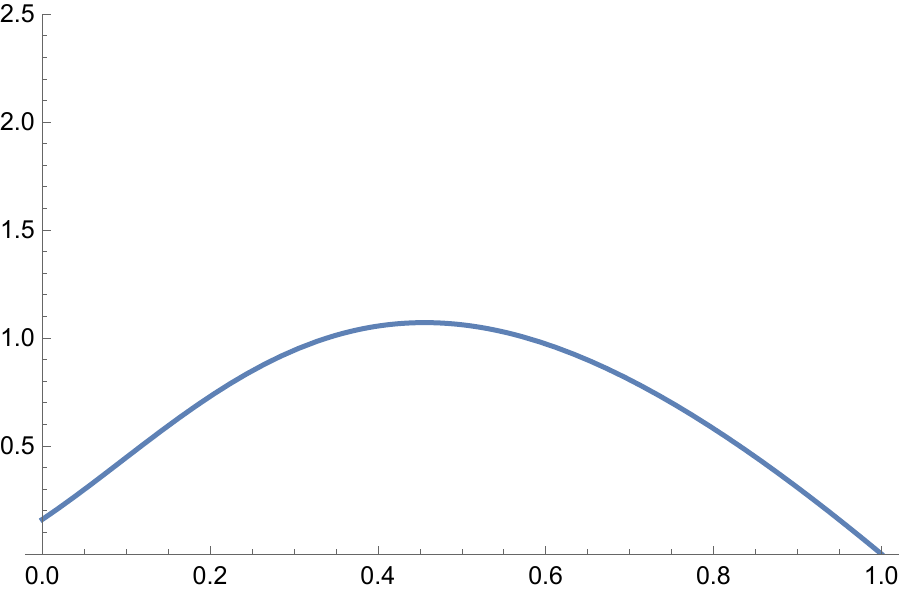} & 
        \includegraphics[width=0.33\textwidth]{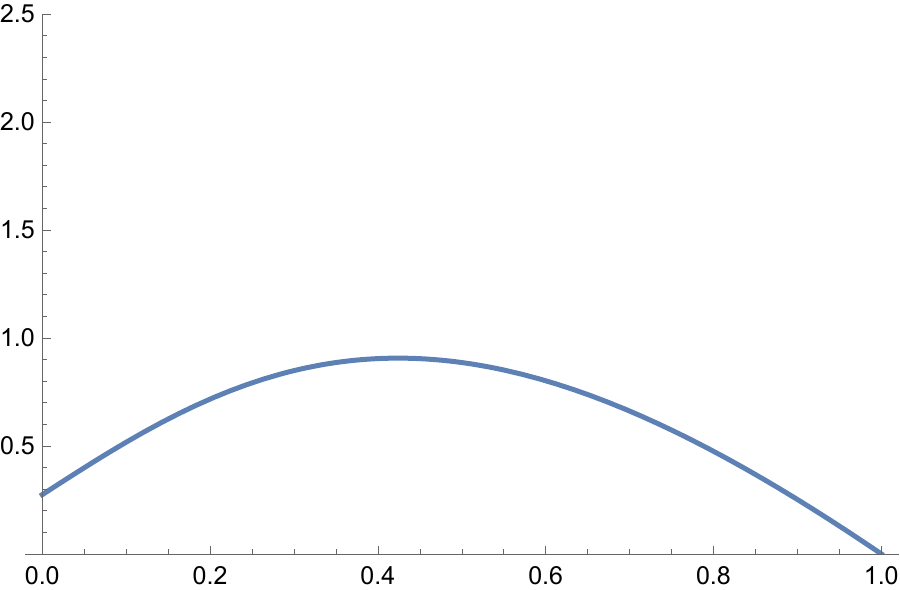} & 
        \includegraphics[width=0.33\textwidth]{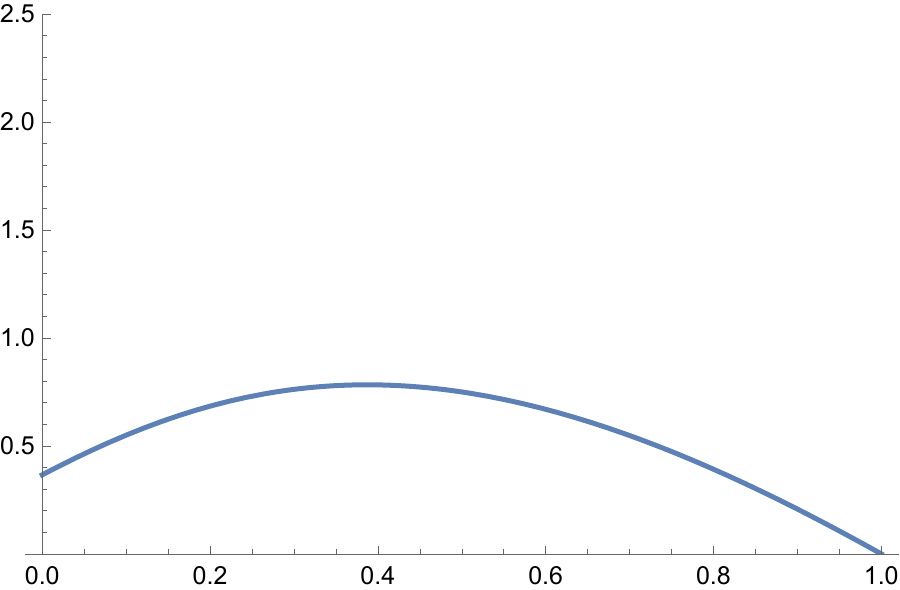} \\
		$T = 0.20$ & $T = 0.25$ & $T = 0.30$ \\ 
    \end{tabular}
    \caption{We plot the transition density $\Gamt(t, x; T, y)$ \eqref{eq:k12-Gamma} for $k=\frac{1}{2}$ model (Section \ref{subsec:k12}) as a function of $y \in [0, 1]$. The fixed parameters are $t=0, x=\frac{1}{2}, L=1, a=1$.}
    \label{fig:gamma}
\end{figure}

\begin{figure}[h!]
		\centering
		\includegraphics[width=0.8\textwidth]{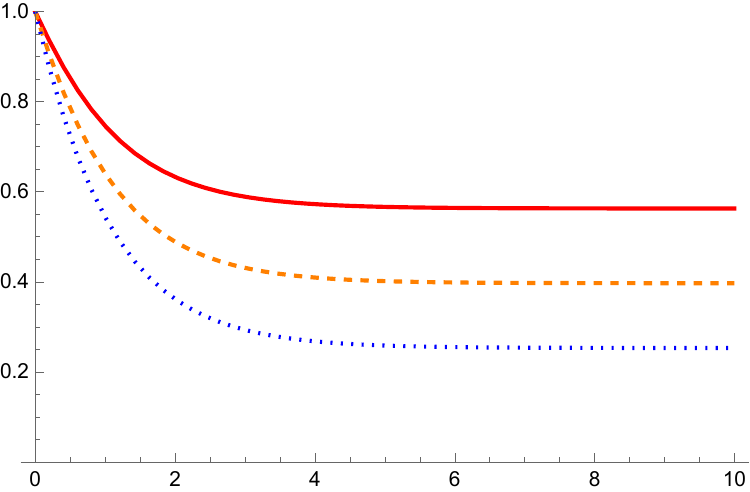}
		\caption{The bond price $B_0^T$ \eqref{eq:k12-bond} for $k=\frac{1}{2}$ model (Section \ref{subsec:k12}). Red solid line corresponds to $x=\frac{1}{3}$, orange dashed line to $x=\frac{1}{2}$, blue dotted line to $x=\frac{2}{3}$. The fixed parameters are $L=1, a=1$.}
		\label{fig:bond}
\end{figure}

\begin{figure}[h!]
		\centering
		\includegraphics[width=0.8\textwidth]{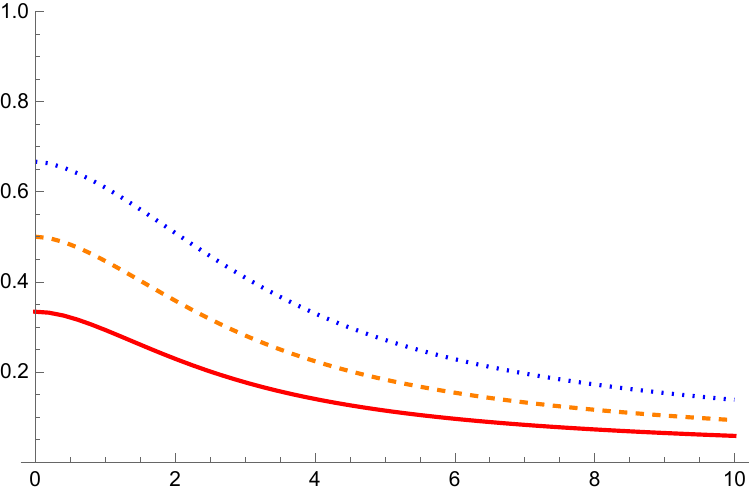}
		\caption{The yield curve $Y_0^T$ \eqref{eq:bond-yield} for $k=\frac{1}{2}$ model (Section \ref{subsec:k12}) with $B_0^T$ as in \eqref{eq:k12-bond}. Red solid line corresponds to $x=\frac{1}{3}$, orange dashed line to $x=\frac{1}{2}$, blue dotted line to $x=\frac{2}{3}$. The fixed parameters are $L=1, a=1$. }
		\label{fig:yield}
\end{figure}

\begin{figure}[h!]
	\begin{tabular}{cc}
		\includegraphics[width=0.5\textwidth]{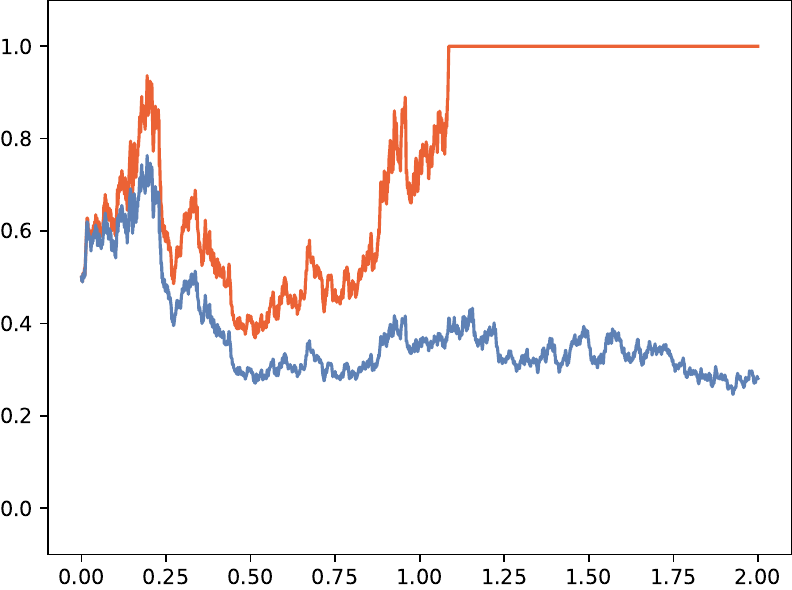} & 
		\includegraphics[width=0.5\textwidth]{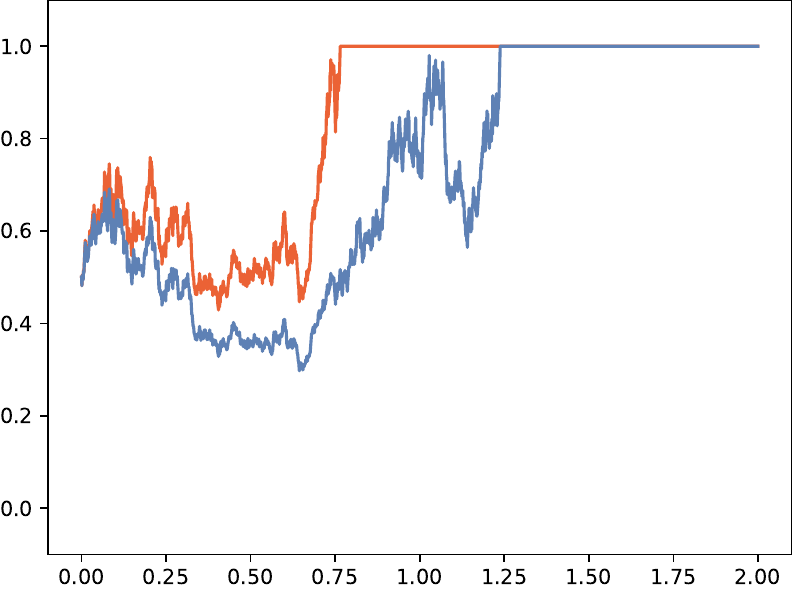} \\
		\includegraphics[width=0.5\textwidth]{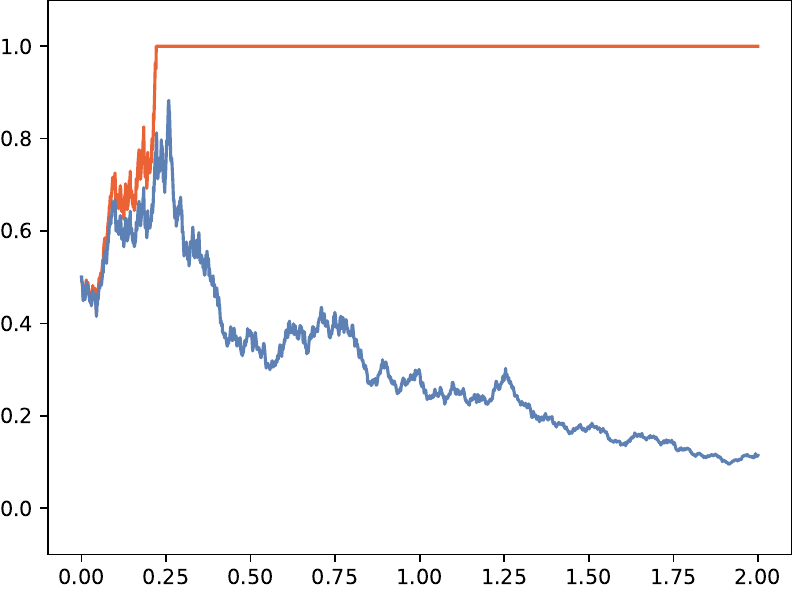} & 
		\includegraphics[width=0.5\textwidth]{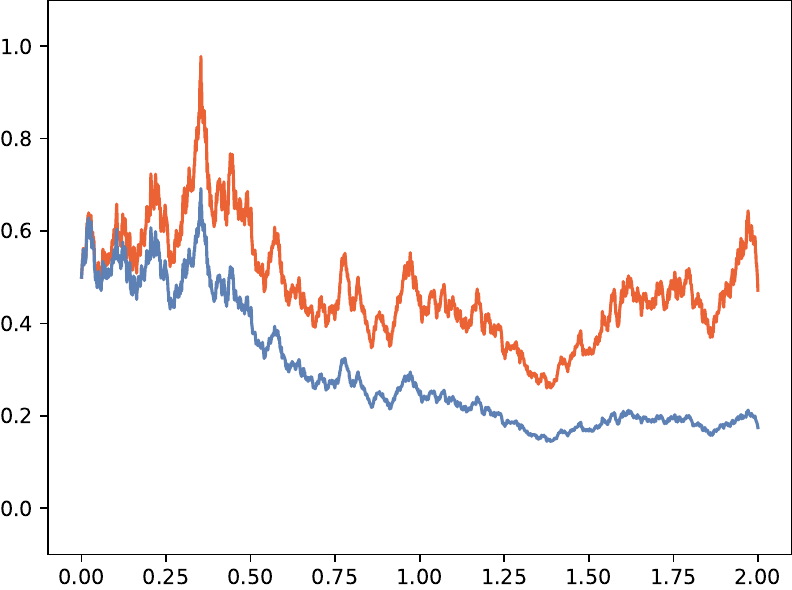} \\
	\end{tabular}
	\centering
	\caption{We plot 4 trajectories for the short-rate $X_t$ in $k=-1/2$ model \ref{subsec:kneg12} with $L=1$, $a=1$, starting from $X_0 = 1/2$, for $t \in [0, 2]$. Blue is under the original probability measure $\Pb$ with dynamics \eqref{eq:kneg12-SDE}, while red is under the equivalent probability measure $\Pbt$ as introduced in Remark \ref{rem:equivprob} and dynamics \eqref{eq:kneg12-equivprob-SDE}.}
    \label{fig:kneg12-trajectories}
\end{figure}

\begin{figure}[h!]
	\includegraphics[width=0.95\textwidth]{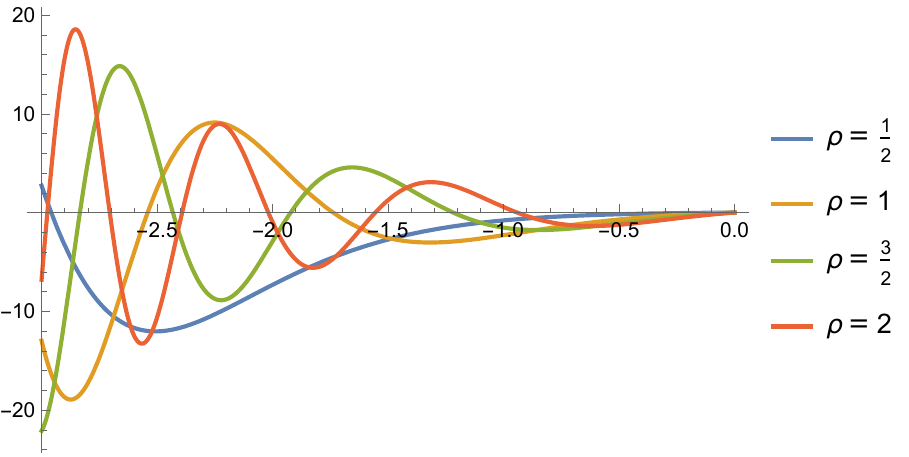}
	\centering 
	\caption{We plot four improper normalized eigenfunctions $\psi(\rho, \ee^z)$ \eqref{eq:kneg12-eigenfunctions} for $k=-1/2$ model with $L =1$, $a = 1$, and $z \in [-3, 0]$ corresponding to $x \in [0.05, 1]$.}
	\label{fig:kneg12-eigenfunctions}
\end{figure}

\begin{figure}[h!]
    \centering
    \begin{tabular}{ccc}
        \includegraphics[width=0.33\textwidth]{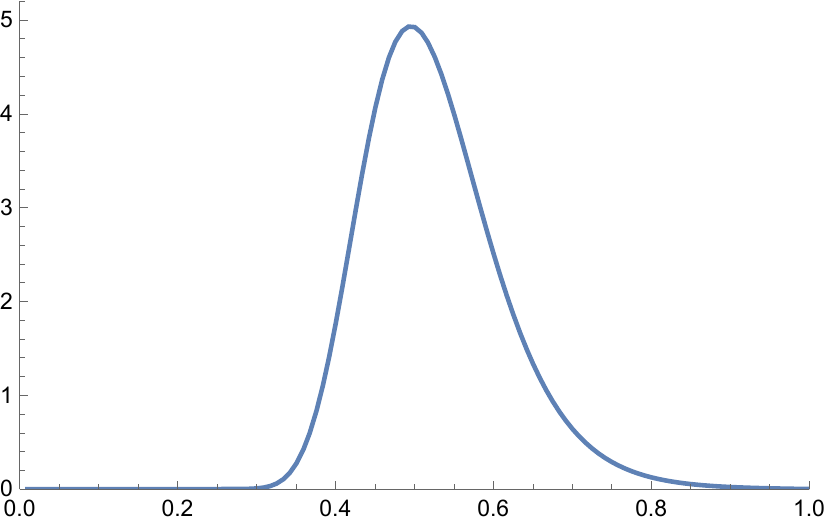} & 
        \includegraphics[width=0.33\textwidth]{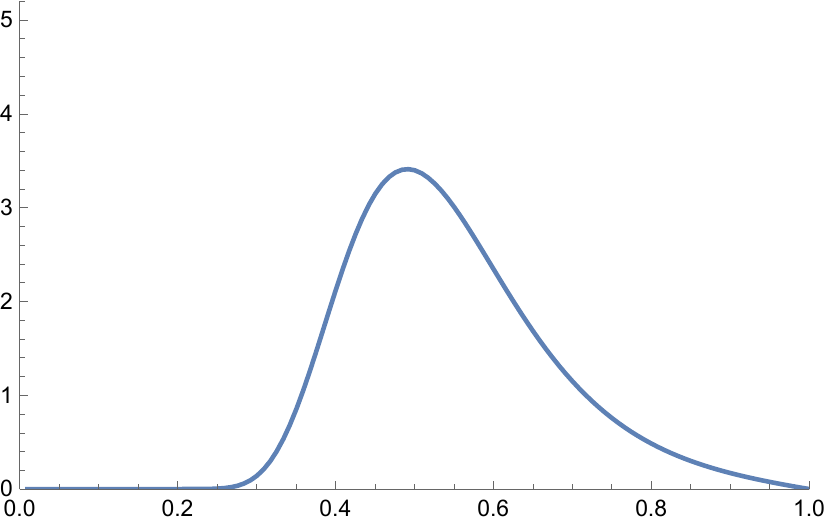} & 
        \includegraphics[width=0.33\textwidth]{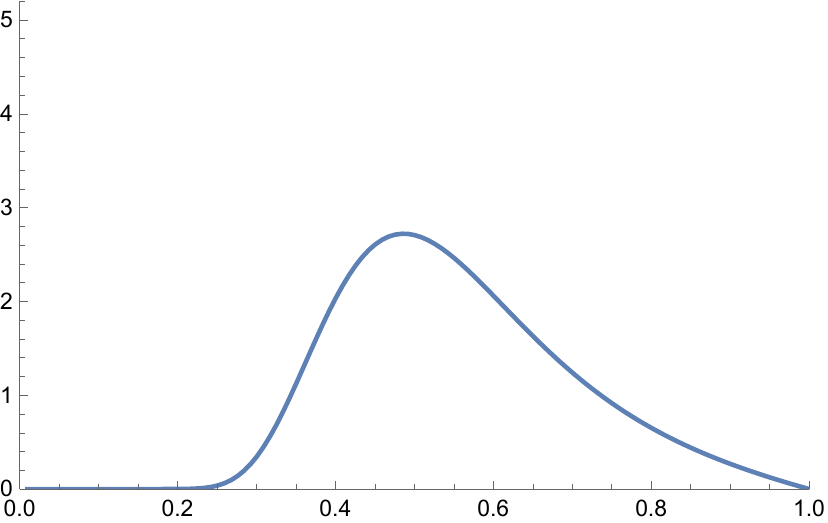} \\
		$T = 0.05$ & $T = 0.10$ & $T = 0.15$ \\ 
        \includegraphics[width=0.33\textwidth]{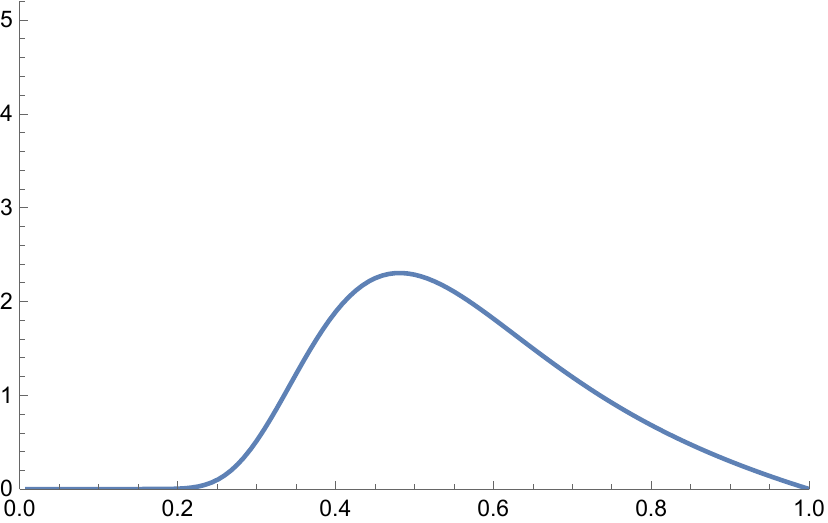} & 
        \includegraphics[width=0.33\textwidth]{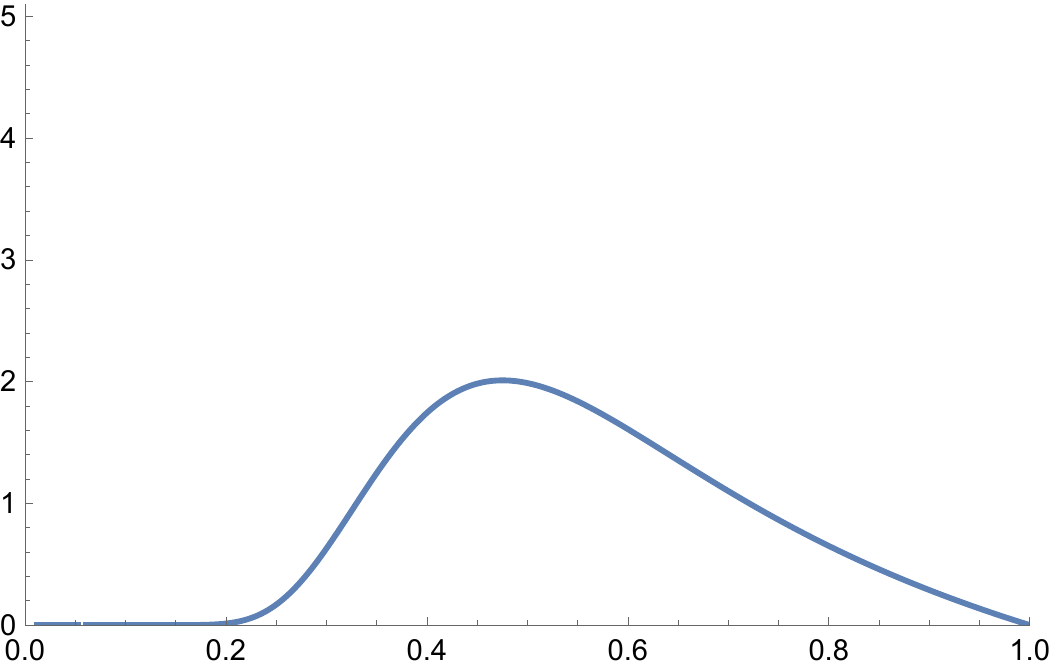} & 
        \includegraphics[width=0.33\textwidth]{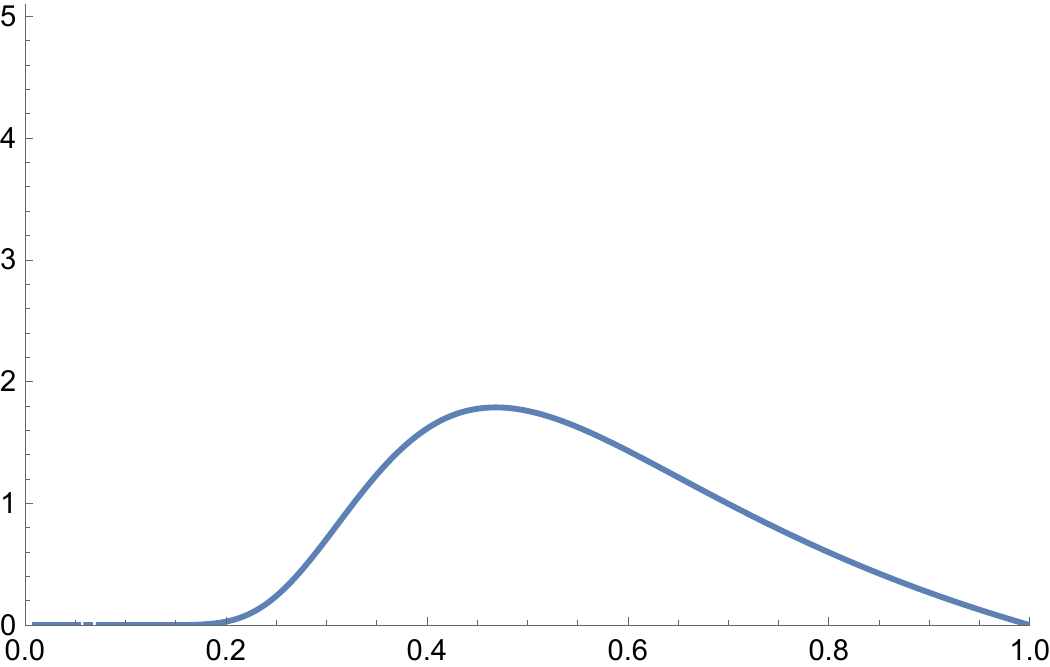} \\
		$T = 0.20$ & $T = 0.25$ & $T = 0.30$ \\ 
    \end{tabular}
    \caption{We plot the transition density $\Gamt(t, x; T, y)$ \eqref{eq:kneg12-Gamma} for $k=-\frac{1}{2}$ model (Section \ref{subsec:kneg12}) as a function of $y \in [0, 1]$ for different values of $T$. The fixed parameters are $t = 0, x = \frac{1}{2}, L = 1, a = 1$.}
    \label{fig:gamma2}
\end{figure}

\begin{figure}[h!]
	\centering
	\includegraphics[width=0.8\textwidth]{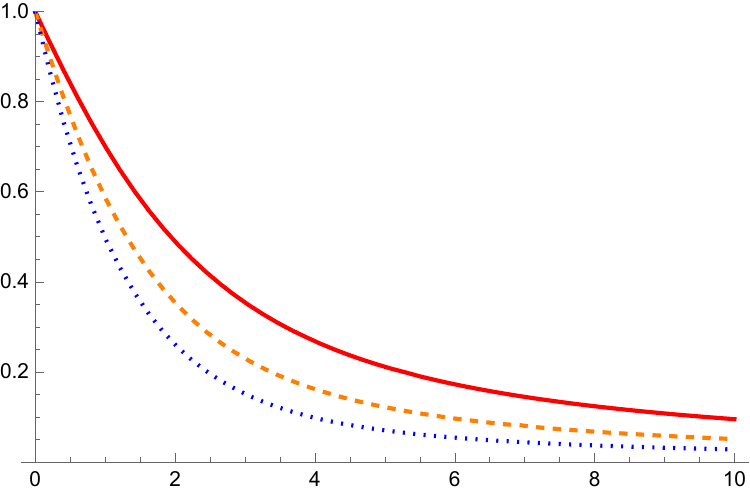}
	\caption{The bond price $B_0^T$ \eqref{eq:kneg12-bond} for $k=-\frac{1}{2}$ model (Section \ref{subsec:kneg12}). Red solid line corresponds to $x=\frac{1}{3}$, orange dashed line to $x=\frac{1}{2}$, blue dotted line to $x=\frac{2}{3}$. The fixed parameters are $L=1, a=1$.}
	\label{fig:bond2}
\end{figure}

\begin{figure}[h!]
	\centering
	\includegraphics[width=0.8\textwidth]{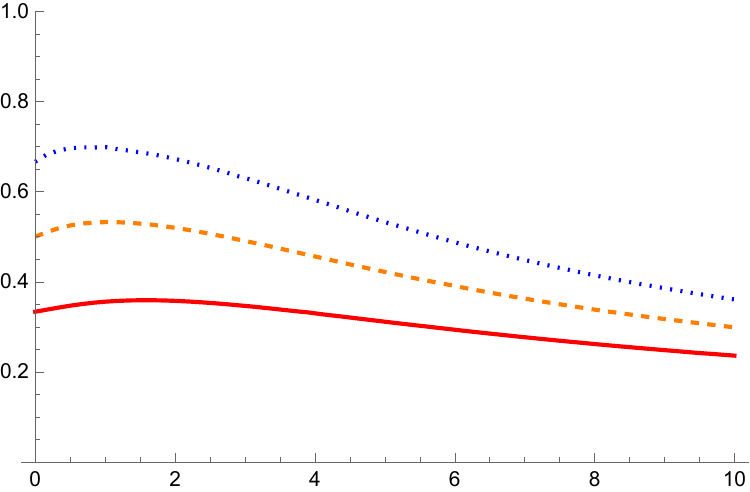}
	\caption{The yield curve $Y_0^T$ \eqref{eq:bond-yield} for $k=-\frac{1}{2}$ model (Section \ref{subsec:kneg12}) with $B_0^T$ as in \eqref{eq:kneg12-bond}. Red solid line corresponds to $x=\frac{1}{3}$, orange dashed line to $x=\frac{1}{2}$, blue dotted line to $x=\frac{2}{3}$. The fixed parameters are $L=1, a=1$.}
	\label{fig:yield2}
\end{figure}

\end{document}